\documentclass[accepted]{uai2026}
                        
\usepackage[american]{babel}

\usepackage{natbib} 
\usepackage{mathtools} 
\usepackage{booktabs} 
\usepackage{tikz} 

\usepackage{amsthm}
\newtheorem{definition}{Definition}

\newtheorem{theorem}{Theorem}[section] 
\newtheorem{corollary}[theorem]{Corollary}
\theoremstyle{remark}
\newtheorem{remark}{Remark} 
\usepackage{comment}

\usepackage{algorithm}

\usepackage{algpseudocode}

\usepackage{amsmath}    
\usepackage{amssymb}    
\usepackage{bm} 
\usepackage{colortbl}
\usepackage{xcolor}

\DeclareMathOperator*{\argmax}{argmax}

\usepackage{amsmath}
\usepackage{amssymb}
\usepackage{amsfonts}
\usepackage{stmaryrd} 
\SetSymbolFont{stmry}{bold}{U}{stmry}{m}{n}
\usepackage{multirow}
\usepackage{bbm}
\usepackage{bm}
\usepackage{enumitem}
\usepackage[dvipsnames]{xcolor}
\usepackage{siunitx}

\newcommand{\lb}[1]{\ensuremath{\mathord{\underline{#1}}}}
\newcommand{\ub}[1]{\ensuremath{\mathord{\overline{#1}}}}
\newcommand{\name}{\textsc{SecureCROWN}}

\newcommand{\tcls}{y}

\title{Privacy-Preserving Robustness Verification for Neural Networks}

\author[1,2]{Nianyun~Song}
\author[3]{Xiaokun~Luan}
\author[1]{Yu~Guo}
\author[1]{Rongfang~Bie}
\author[3]{Meng~Sun}
\author[2]{Xiyue~Zhang\thanks{Correspondence to XZ (xiyue.zhang@bristol.ac.uk)}}

\affil[1]{%
    School of Artificial Intelligence, Beijing Key Laboratory of Artificial Intelligence for Education, Engineering Research Center of Intelligent Technology and Educational Application (Ministry of Education), Beijing Normal University, \protect\\
    Beijing, China
}

\affil[2]{%
    School of Computer Science\\
    University of Bristol\\
    Bristol, UK
}
\affil[3]{%
    School of Mathematical Sciences\\
    Peking University\\
    Beijing, China
}

\begin{document}
\maketitle

\begin{abstract}
Neural network verification and data privacy are inherently in tension: \emph{verification demands full access to model parameters and input data, yet both are increasingly restricted by privacy regulations and intellectual property constraints}. This tension has left robustness verification impractical in privacy-sensitive domains.
In this work, we address this gap with \name{}, the first framework for privacy-preserving neural network robustness verification. Built upon secure two-party computation (2PC), our framework enables a model owner and a data owner to jointly compute certified robustness bounds---revealing only the final result while provably protecting both parties' private data under the semi-honest security model.
A key challenge is securely computing the conditional operations in Linear Bound Propagation, where the data-dependent branching is incompatible with standard secure computation protocols.
We eliminate branching by formulating conditional logic as 
continuous arithmetic operations.
Additionally, we introduce a Newton--Raphson refinement method to improve numerical stability.
Extensive analysis and experiments show that \name{} strictly matches plaintext verification results, while completing in 0.1--200s across varied model sizes and communication settings (LAN/WAN), demonstrating the feasibility of privacy-preserving neural network verification.

\end{abstract}

\section{Introduction}
\label{sec:introduction}

The field of deep learning (DL) has experienced enormous growth recently,
with deep neural networks (DNNs) now deployed in a variety of safety-critical applications, including medical diagnosis~\citep{litjens2017survey,esteva2017dermatologist} and healthcare~\citep{esteva2019guide,davenport2019potential}. In these domains, safety and robustness of DNNs are essential. 
However, DNNs are susceptible to perturbations in their inputs~\citep{Szegedy13intriguing}, such as noise and illumination variations~\citep{Goodfellow2014explain,Hendrycks2019Benchmark}. 
To ensure DNN robustness with worst-case guarantees, significant progress in verification approaches~\citep{Wu2024Marabou,wang2021beta,singh2019abstract} has been made, which can certify that a model's prediction remains stable regardless of data noise or environmental changes.

Despite these advances, existing verification techniques assume centralized access to both model parameters and input data in plaintext.
This assumption is often unrealistic in privacy-sensitive deployments.
In practice, user data may be subject to strict privacy regulations, such as the General Data Protection Regulation (GDPR)~\citep{voigt2017eu}. 
Meanwhile, proprietary model parameters constitute valuable intellectual property and may not be disclosed to third parties. 
Thus, the standard verification setting, where both model and data are fully accessible, is incompatible with many privacy-regulated collaborative scenarios.

We consider a two-party verification setting in which a model owner $P_0$ holds a proprietary DNN, while a data owner $P_1$ holds sensitive input data. We aim to verify the robustness of the model on the given input such that only the verification result is revealed, while both the model parameters and the input data remain private. Achieving this requires jointly performing robustness verification without exposing either party's confidential information. 

Secure multi-party computation (MPC)~\citep{Evans2018MPC} provides a principled cryptographic approach that enables multiple parties to jointly evaluate functions on their private inputs without revealing their data to one another.
While MPC has been successfully applied to privacy-preserving inference and training~\citep{Knott2021CrypTen,Tan2021cryptgpu,Wagh2021falcon,Kumar2020cryptflow,Gupta2025shark,Zhang2021Survey}, extending it to robustness verification faces two fundamental challenges:
\emph{(i) Data-dependent branching:} Convex relaxation of ReLU activations requires distinct parameters depending on secret-shared activation states. In MPC, resolving these conditionals requires a secure comparison for every neuron, scaling linearly with the total number of neurons and rendering naïve implementations computationally prohibitive.
\emph{(ii) Precision sensitivity in deep recursion:} Unlike MPC-based inference, which tolerates minor numerical errors during a single forward pass, verification computes certified bounds via deep backward recursion that amplifies numerical errors. Because safety guarantees depend strictly on the sign of the final margin, these compounded errors can invert near-zero bounds, compromising verification soundness.

To bridge this gap, we propose \name{}, the first privacy-preserving framework for neural network robustness verification built on secure two-party computation (2PC). 
It enables verification of DNNs while ensuring that sensitive input data remains confidential and proprietary model parameters are not disclosed.
Specifically, we develop \textit{(i)} a secure implementation of Linear Bound Propagation (LBP) that computes certified robustness bounds under privacy constraints; \textit{(ii)} a unified branch-free reformulation of the conditional logic in ReLU relaxation and bound propagation, efficiently realized via function secret sharing (FSS); and \textit{(iii)} a Newton–Raphson-based refinement method for reciprocal computation to improve numerical precision and stability.
Furthermore, we provide formal security guarantees in the semi-honest adversary model and conduct extensive experiments validating that our protocol achieves verification precision comparable to the plaintext verifier while incurring moderate cryptographic overhead.

\section{Related Work}
Modern neural network verification methods primarily address robustness analysis in a plaintext setting, where both model parameters and user inputs are assumed to be fully accessible.
Representative verifiers include $\alpha,\beta$-CROWN~\citep{zhang2018efficient,xu2020automatic,xu2021fast,wang2021beta}, CORA~\citep{Althoff15CORA}, Marabou~\citep{Wu2024Marabou,Katz2019Marabou}, and NeuralSAT~\citep{Duong2024neuralsat,Duong2025NeuralSAT}. 
These approaches rely on techniques such as convex relaxation~\citep{salman2019convex}, abstract interpretation~\citep{singh2019abstract}, and constraint-solving~\citep{Hai2023neuralsat} to compute certified robustness bounds or to solve verification queries exactly. 

Meanwhile, substantial progress has been made in privacy-preserving machine learning through cryptographic protocols. A prominent line of work employs MPC to enable secure inference.
CrypTen~\citep{Knott2021CrypTen} supports private inference over secret-shared inputs and models in a two-party setting, assuming a semi-honest threat model.
CryptGPU~\citep{Tan2021cryptgpu} extends this framework to a three-party system with optimized GPU-based protocols under semi-honest security.
Falcon~\citep{Wagh2021falcon} implements strengthened malicious security in a three-party setting and combines SecureNN~\citep{wagh2019securenn} and ABY$^{3}$~\citep{mohassel2018aby3} to improve protocol efficiency. 
ObliviGate~\citep{obligate} further achieves maliciously secure inference while concealing the model architecture.
CrypTFlow~\citep{Kumar2020cryptflow} proposes an end-to-end toolchain that compiles predefined TensorFlow inference code into MPC protocols, supporting both semi-honest and malicious adversaries.
Beyond MPC, zero-knowledge proof systems have been leveraged to provide verifiable inference without revealing model parameters or inputs~\citep{Lee2024zk,Maheri2025telesparse}, while homomorphic encryption (HE) enables neural network inference directly over encrypted data~\citep{Lee2022HE,Lou2021HEMET}.

However, these cryptographic methods focus exclusively on secure or verifiable inference. They do not address robustness verification, which requires computing and propagating symbolic bounds beyond forward inference alone. 
In contrast, we establish a privacy-preserving framework for neural network robustness verification. To the best of our knowledge, our work is the first to integrate formal robustness verification with cryptographic privacy guarantees.

\section{Preliminaries}
\label{sec:preliminaries}

\subsection{Local Robustness Verification}
\label{subsec:nn_verification}
Consider an $L$-layer feedforward neural network classifier $f: \mathbb{R}^{d_0} \rightarrow \mathbb{R}^{K}$, where $d_0$ is the input dimension and $K$ is the number of classes. The network is parameterized by weight matrices $\mathbf{W}^{(l)} \in \mathbb{R}^{d_l \times d_{l-1}}$ and bias vectors $\mathbf{b}^{(l)} \in \mathbb{R}^{d_l}$ for layers $l \in \{1, \ldots, L\}$, with $d_L = K$. The pre-activation at each layer is defined as $\mathbf{z}^{(1)} = \mathbf{W}^{(1)}\mathbf{x} + \mathbf{b}^{(1)}$ for $l=1$, and $\mathbf{z}^{(l)} = \mathbf{W}^{(l)}\sigma(\mathbf{z}^{(l-1)}) + \mathbf{b}^{(l)}$ for $l \geq 2$.
 In this work, we consider the ReLU activation $\sigma(\cdot) = \max(0, \cdot)$.

The network output is $f(\mathbf{x}) = \mathbf{z}^{(L)}$, with predicted class $\tcls = \argmax_{k} f_k(\mathbf{x})$.
Given an input $\mathbf{x}_0$ correctly classified as class $\tcls$, the network is \emph{locally robust} within an $\ell_p$-ball of radius $\epsilon$ if:
\begin{equation}\label{eq:robustness}
    {f}_{\tcls,j} \triangleq \min_{\mathbf{x} \in \mathcal{B}_p(\mathbf{x}_0, \epsilon)} \bigl(f_\tcls(\mathbf{x}) - f_j(\mathbf{x})\bigr) > 0, \quad \forall\, j \neq \tcls,
\end{equation}
where $\mathcal{B}_p(\mathbf{x}_0, \epsilon) = \{\mathbf{x} : \|\mathbf{x} - \mathbf{x}_0\|_p \leq \epsilon\}$. Since exactly verifying Eq.~\eqref{eq:robustness} is NP-hard for ReLU networks~\citep{katz2017reluplex}, verification methods typically compute a certified lower bound $\underline{f}_{\tcls,j}$ via linear relaxation~\citep{salman2019convex}.

\textbf{Linear Relaxation.}
Consider robustness verification under $\ell_p$ input perturbations $\mathbf{x} \in \mathcal{B}_p(\mathbf{x}_0, \epsilon)$.
Let $\lb{z}, \ub{z}\in\mathbb{R}$ be the pre-activation lower and upper bounds for a neuron $z$, i.e., $\lb{z} \leq z \leq \ub{z}$. 
To handle the nonlinearity of the activation function, linear relaxation constructs a pair of linear bounding functions—a lower bound $\lb{\alpha} z + \lb{\beta}$ and an upper bound $\ub{\alpha} z + \ub{\beta}$—satisfying $\lb{\alpha} z + \lb{\beta} \leq \sigma(z) \leq \ub{\alpha} z + \ub{\beta}$ for all $z \in [\lb{z}, \ub{z}]$.
A commonly used bounding strategy for ReLU in CROWN~\citep{zhang2018efficient} sets the upper-bound parameters $(\ub{\alpha}, \ub{\beta})$ as:
\begin{equation}\label{eq:relu_relax}
(\ub{\alpha}, \ub{\beta}) =
\begin{cases}
(1,0), & \text{if } \lb{z} \ge 0, \\
(0,0), & \text{if } \ub{z} \le 0, \\
\left(\alpha^*, -\alpha^* \lb{z}\right), 
& \text{otherwise},
\end{cases}
\quad
\text{where }
\alpha^* = \tfrac{\ub{z}}{\ub{z}-\lb{z}},
\end{equation}
and the lower bound shares the same slope with zero intercept: $(\lb{\alpha}, \lb{\beta}) = (\ub{\alpha}, 0)$.

\textbf{Linear Bound Propagation.}
LBP in CROWN-series verifiers computes certified bounds by propagating backward linear relaxations through the network.
To bound $\mathbf{z}^{(l)}$ for any layer $l \in \{1, \ldots, L\}$, LBP constructs linear functions such that
\begin{equation}\label{eq:general_lb}
\mathbf{A}^{(0)} \mathbf{x} + \underline{\mathbf{c}}^{(0)} \leq \mathbf{z}^{(l)} \leq \mathbf{A}^{(0)} \mathbf{x} + \bar{\mathbf{c}}^{(0)}, \ \ \forall \mathbf{x} \in \mathcal{B}_p(\mathbf{x}_0, \epsilon),
\end{equation}
where $\mathbf{A}^{(0)} \in \mathbb{R}^{d_l \times d_0}$ and $\underline{\mathbf{c}}^{(0)}, \bar{\mathbf{c}}^{(0)} \in \mathbb{R}^{d_l}$. Since our relaxation uses identical slopes $\underline{\alpha} = \bar{\alpha}$, the coefficient matrix $\mathbf{A}^{(0)}$ is shared for both bounds, while distinct intercepts $\underline{\beta} \neq \bar{\beta}$ yield separate constants $\underline{\mathbf{c}}^{(0)}$ and $\bar{\mathbf{c}}^{(0)}$. 

For $l = 1$, since $\mathbf{z}^{(1)} = \mathbf{W}^{(1)} \mathbf{x} + \mathbf{b}^{(1)}$ is purely an affine transformation of $\mathbf{x}$, we directly obtain $\mathbf{A}^{(0)} = \mathbf{W}^{(1)}$ and $\underline{\mathbf{c}}^{(0)} = \bar{\mathbf{c}}^{(0)} = \mathbf{b}^{(1)}$. For $l \geq 2$, the backward recursion requires relaxation parameters $(\boldsymbol{\alpha}^{(t)}, \underline{\boldsymbol{\beta}}^{(t)}, \bar{\boldsymbol{\beta}}^{(t)})$ for each intermediate layer $t \in \{1, \ldots, l-1\}$, where $\boldsymbol{\alpha}^{(t)} \in \mathbb{R}^{d_t}$ denotes slopes and $\underline{\boldsymbol{\beta}}^{(t)}, \bar{\boldsymbol{\beta}}^{(t)} \in \mathbb{R}^{d_t}$ intercepts; these are determined by the pre-activation bounds via Eq.~\eqref{eq:relu_relax}. 
The recursion maintains $\mathbf{A}^{(t)} \in \mathbb{R}^{d_l \times d_t}$ and $\underline{\mathbf{c}}^{(t)}, \bar{\mathbf{c}}^{(t)} \in \mathbb{R}^{d_l}$ such that $\mathbf{A}^{(t)} \mathbf{z}^{(t)} + \underline{\mathbf{c}}^{(t)} \leq \mathbf{z}^{(l)} \leq \mathbf{A}^{(t)} \mathbf{z}^{(t)} + \bar{\mathbf{c}}^{(t)}$. Starting with $\mathbf{A}^{(l)} = \mathbf{I}_{d_l}$ and $\underline{\mathbf{c}}^{(l)} = \bar{\mathbf{c}}^{(l)} = \mathbf{0}$, the backward updates for $t = l, \ldots, 2$ are given by:
\begin{align}
\hat{\mathbf{A}}^{(t)} &= \mathbf{A}^{(t)} \mathbf{W}^{(t)}, \label{eq:Ahat_update} \\
\mathbf{A}^{(t-1)} &= \hat{\mathbf{A}}^{(t)} \operatorname{diag}(\boldsymbol{\alpha}^{(t-1)}), \label{eq:A_update} \\
\underline{\mathbf{c}}^{(t-1)} &= \underline{\mathbf{c}}^{(t)} + \mathbf{A}^{(t)} \mathbf{b}^{(t)} + \underline{\boldsymbol{\delta}}^{(t-1)}, \label{eq:c_lb_update} \\
\bar{\mathbf{c}}^{(t-1)} &= \bar{\mathbf{c}}^{(t)} + \mathbf{A}^{(t)} \mathbf{b}^{(t)} + \bar{\boldsymbol{\delta}}^{(t-1)}, \label{eq:c_ub_update}
\end{align}
where the $i$-th element (for $i \in \{1, \ldots, d_l\}$) of the intercept contribution is $\underline{\delta}^{(t-1)}_i = \sum_{j:\, \hat{A}^{(t)}_{ij} > 0} \hat{A}^{(t)}_{ij} \underline{\beta}^{(t-1)}_j + \sum_{j:\, \hat{A}^{(t)}_{ij} \leq 0} \hat{A}^{(t)}_{ij} \bar{\beta}^{(t-1)}_j$, and $\bar{\boldsymbol{\delta}}^{(t-1)}$ is defined analogously with $\underline{\beta}$ and $\bar{\beta}$ swapped. The recursion terminates at the input layer yielding $\mathbf{A}^{(0)} = \mathbf{A}^{(1)} \mathbf{W}^{(1)}$, $\underline{\mathbf{c}}^{(0)} = \underline{\mathbf{c}}^{(1)} + \mathbf{A}^{(1)} \mathbf{b}^{(1)}$, and $\bar{\mathbf{c}}^{(0)} = \bar{\mathbf{c}}^{(1)} + \mathbf{A}^{(1)} \mathbf{b}^{(1)}$. 
By applying H\"older's inequality over the $\ell_p$-ball $\mathcal{B}_p(\mathbf{x}_0, \epsilon)$ (where $1/p + 1/q = 1$), the symbolic bounds in Eq.~\eqref{eq:general_lb} are converted into concrete element-wise bounds:
\begin{align}
\underline{z}^{(l)}_i &= \mathbf{A}^{(0)}_{i,:} \mathbf{x}_0 + \underline{c}^{(0)}_i - \epsilon \|\mathbf{A}^{(0)}_{i,:}\|_q, \label{eq:z_lb} \\
\bar{z}^{(l)}_i &= \mathbf{A}^{(0)}_{i,:} \mathbf{x}_0 + \bar{c}^{(0)}_i + \epsilon \|\mathbf{A}^{(0)}_{i,:}\|_q. \label{eq:z_ub}
\end{align}

\textit{Margin verification.}
For margin verification between classes $\tcls$ and $j$, we apply LBP to the output layer ($l = L$) by setting $\mathbf{W}^{(L)} \leftarrow \mathbf{W}^{(L)}_{\tcls,:} - \mathbf{W}^{(L)}_{j,:} \in \mathbb{R}^{1 \times d_{L-1}}$ and $b^{(L)} \leftarrow b^{(L)}_{\tcls} - b^{(L)}_{j} \in \mathbb{R}$. With $d_L = 1$, the coefficient matrix $\mathbf{A}^{(t)}$ reduces to a row vector and $\underline{c}^{(t)}$ to a scalar. The certified lower bound is:
\begin{equation}\label{eq:final_lb}
\underline{f}_{\tcls,j} = \mathbf{A}^{(0)} \mathbf{x}_0 + \underline{c}^{(0)} - \epsilon \|\mathbf{A}^{(0)}\|_q.
\end{equation}
Local robustness is verified if $\underline{f}_{\tcls,j} > 0$.

\subsection{Secure Computation Primitives}
\label{subsec:2pc}

In the standard 2PC setting~\citep{yao1982protocols,goldreich1987how}, two parties ($P_0$ and $P_1$) jointly compute a function over their private inputs, revealing nothing beyond the prescribed output. Computations are performed in the ring $\mathbb{Z}_{2^k}$ of bit-length $k$ (modulo $2^k$ is omitted for brevity). A real value $v \in \mathbb{R}$ is encoded in two's complement fixed-point format with $n_f$ fractional bits as $\tilde{v} = \lfloor v \cdot 2^{n_f} \rfloor \in \mathbb{Z}_{2^k}$.

\textbf{Additive Secret Sharing (ASS).}
A value $x\in\mathbb{Z}_{2^k}$ is additively shared as 
$\langle x \rangle = \{\langle x \rangle_0, \langle x \rangle_1\}$ 
such that $x = \langle x\rangle_0+\langle x\rangle_1$.
To generate shares, $P_0$ samples $\langle x\rangle_0 \in \mathbb{Z}_{2^k}$ 
uniformly at random and sets $\langle x\rangle_1 = x-\langle x\rangle_0$; 
each individual share is information-theoretically independent of $x$.
ASS supports linear operations locally. Given shares $\langle x \rangle, 
\langle y \rangle$ and a public value $c\in\mathbb{Z}_{2^k}$, party $P_i$ 
($i \in \{0,1\}$) computes
\begin{equation*}
    \langle x + y \rangle_i = \langle x \rangle_i + \langle y \rangle_i, 
    \quad \langle c \cdot x \rangle_i = c \cdot \langle x \rangle_i.
\end{equation*}
For adding a public constant, one party (e.g., $P_0$) absorbs it locally: 
$\langle x+c\rangle_0=\langle x\rangle_0+c$ and 
$\langle x+c\rangle_1=\langle x\rangle_1$.

\textbf{Secure Multiplication via Beaver Triples.}
Nonlinear operations require interaction. A standard technique uses a Beaver triple $(\langle a \rangle, \langle b \rangle, \langle c \rangle)$ with $c = a \cdot b$. Given $\langle x\rangle,\langle y\rangle$, parties compute masked differences $e = x - a$ and $d = y - b$ by locally forming $\langle e\rangle=\langle x\rangle-\langle a\rangle$ and $\langle d\rangle=\langle y\rangle-\langle b\rangle$, then opening $e$ and $d$ (i.e., exchanging shares and reconstructing the public values).
They then output product shares
\begin{equation*}
     \langle x \cdot y \rangle_i = \langle c \rangle_i + d \cdot \langle a \rangle_i + e \cdot \langle b \rangle_i + (1-i) \cdot e \cdot d,
\end{equation*}
where $i \in \{0, 1\}$.
This yields $\langle x\cdot y\rangle_0+\langle x\cdot y\rangle_1 = xy$, followed by a truncation procedure to restore the fixed-point scale, i.e., $\langle \widetilde{xy} \rangle = \langle \lfloor \tilde{x}\cdot \tilde{y}/ 2^{n_f} \rfloor \rangle$.

\textbf{Function Secret Sharing.}
FSS generalizes secret sharing from values to functions~\citep{BoyleGI15}. A two-party FSS scheme for a function family $\mathcal{G}$ consists of:
\begin{itemize}[leftmargin=15pt, nosep]
    \item $\mathsf{Gen}(1^\lambda, g) \to (k_0, k_1)$: a randomized key generation algorithm that, given a security parameter $\lambda$ and a function $g\in\mathcal{G}$, outputs two correlated keys $(k_0,k_1)$.
    \item $\mathsf{Eval}(i, k_i, x) \to y_i$: a deterministic evaluation algorithm run by $P_i$ on input $x$ (typically a public masked value in MPC), producing an output share $y_i$ such that $y_0 + y_1 = g(x)$ (over an appropriate output ring).
\end{itemize}
Intuitively, each key $k_i$ hides the function $g$, yet the two evaluations add up to the correct value.
A key instantiation is the Distributed Comparison Function (DCF), computing $g^<_{\theta, \rho}(x) = \rho \cdot \mathbb{I}\{x < \theta\}$ for threshold $\theta$ and payload $\rho$, where $\mathbb{I}$ is the indicator function. DCF is a basic primitive for secure comparison and is commonly used to build higher-level operators such as ReLU and fixed-point truncation.

\textbf{FSS-Based Secure Evaluation on Masked Values.}
Many FSS-based 2PC protocols maintain a \emph{masking invariant} on intermediate values. Consequently, for a private value $x\in\mathbb{Z}_{2^k}$, the parties hold a public masked value $\hat{x} = x + r_{\mathrm{in}}$, where $r_{\mathrm{in}}$ is a preprocessed random mask unknown to either party. To evaluate $y=g(x)$ while preserving the invariant, the parties generate FSS keys for a shifted function:
\begin{equation*}
    g_{r_{\mathrm{in}}, r_{\mathrm{out}}}(\hat{x}) = g(\hat{x} - r_{\mathrm{in}}) + r_{\mathrm{out}},
\end{equation*}
where $r_{\mathrm{out}}$ is a fresh output mask. Each party evaluates its key on $\hat{x}$ to obtain shares of $g_{r_{\mathrm{in}},r_{\mathrm{out}}}(\hat{x})$, which reconstruct to $\hat{y} = g(x) + r_{\mathrm{out}}$.
Thus the output remains masked in the same form, enabling secure composition of function evaluations without revealing intermediate plaintexts.

\section{Problem Formulation}
\label{sec:problem}

\textbf{System Setup.}
We study privacy-preserving robustness verification in a 2PC setting with a \emph{model owner} ($P_0$) and a \emph{data owner} ($P_1$). The network architecture (number of layers $L$ and layer dimensions $\{d_l\}_{l=0}^{L}$) is assumed to be public. The model owner $P_0$ holds a trained neural network $f$ with private parameters $\{\mathbf{W}^{(l)}, \mathbf{b}^{(l)}\}_{l=1}^{L}$, which represent proprietary intellectual property.
The data owner $P_1$ holds a verification query consisting of input $\mathbf{x}_0 \in \mathbb{R}^{d_0}$, {robustness} perturbation radius $\epsilon > 0$
{(the magnitude of worst-case perturbations to certify against)}, and the classification target pair (true label $\tcls$, target label $j$). During setup, $P_0$ secret-shares the model parameters, retains $\{\langle \mathbf{W}^{(l)} \rangle_0, \langle \mathbf{b}^{(l)} \rangle_0\}_{l=1}^{L}$, and sends $\{\langle \mathbf{W}^{(l)} \rangle_1, \langle \mathbf{b}^{(l)} \rangle_1\}_{l=1}^{L}$ to $P_1$. 
To avoid the $O(K)$ overhead associated with secure array indexing, $P_1$ encodes the verification target as a difference vector $\mathbf{d} = \mathbf{e}_\tcls - \mathbf{e}_j \in \{-1, 0, 1\}^K$ (where $\mathbf{e}_\tcls$ and $\mathbf{e}_j$ are one-hot vectors). $P_1$ then secret-shares these encoded inputs, retaining $(\langle \mathbf{x}_0 \rangle_1, \langle \epsilon \rangle_1, \langle \mathbf{d} \rangle_1)$ and sending $(\langle \mathbf{x}_0 \rangle_0, \langle \epsilon \rangle_0, \langle \mathbf{d} \rangle_0)$ to $P_0$.

\textbf{Threat Model.} 
We consider the \emph{semi-honest} (honest-but-curious) adversarial model, where both parties follow the protocol specification but may attempt to infer the other party's private input from their view. The protocol operates in a preprocessing model with a trusted dealer $\mathcal{D}$ that distributes correlated randomness (Beaver triples, FSS keys) during an offline phase. The preprocessing depends only on public parameters, so $\mathcal{D}$ learns nothing about private inputs. We assume $\mathcal{D}$ does not collude with either party. In practice, $\mathcal{D}$ can be instantiated via trusted hardware or replaced by a two-party preprocessing protocol.

\textbf{Problem Statement.}
Given secret shares of a neural network $f$ and a verification query $(\mathbf{x}_0, \epsilon, \mathbf{d})$ held by the two parties (where $\tcls$ and $j$ are encoded in the difference vector $\mathbf{d}$), design a 2PC protocol that enables $P_1$ to determine whether the network is certifiably $\ell_\infty$-robust at $(\mathbf{x}_0, \epsilon)$ for the margin between classes $\tcls$ and $j$. The parties jointly compute secret shares of the certified lower bound $\underline{f}_{\tcls,j}$ (Eq.~\eqref{eq:final_lb}). Upon completion, $P_0$ sends $\langle \underline{f}_{\tcls,j} \rangle_0$ to $P_1$, who reconstructs $\underline{f}_{\tcls,j} = \langle \underline{f}_{\tcls,j} \rangle_0 + \langle \underline{f}_{\tcls,j} \rangle_1$ and concludes robust if $\underline{f}_{\tcls,j} > 0$, {or} unknown otherwise.

While the problem formulation above is general, our current protocol instantiates it for fully connected ReLU networks; extensions to other architectures are discussed in Appendix~\ref{app:extensions}.

\textbf{Design Goals.}
\name{} targets two objectives:
\textbf{(1) Input and Model Privacy:} The protocol satisfies semi-honest security. $P_0$ learns nothing about $(\mathbf{x}_0, \epsilon, \mathbf{d})$ beyond public parameters, and $P_1$ learns only $\underline{f}_{\tcls,j}$. All intermediate quantities remain secret-shared throughout.
\textbf{(2) Faithful Verification Semantics:} The protocol preserves the decision behavior of the plaintext verifier (Section~\ref{subsec:nn_verification}). Any numerical error introduced by fixed-point arithmetic must be bounded to prevent unsound certifications.

\textbf{Security Definition.} 
We formalize semi-honest security in the real/ideal paradigm. Let $\mathsf{inp}_0 = \{\mathbf{W}^{(l)}, \mathbf{b}^{(l)}\}_{l=1}^{L}$ and $\mathsf{inp}_1 = (\mathbf{x}_0, \epsilon, \mathbf{d})$ denote the private inputs of $P_0$ and $P_1$, respectively. The ideal functionality is defined as $\mathcal{F}(\mathsf{inp}_0, \mathsf{inp}_1) = (\bot, \underline{f}_{\tcls,j})$, where $P_0$ receives no output ($\bot$) and $P_1$ receives the certified lower bound.

\begin{definition}[Semi-Honest Security]
\label{def:semihonest_security}
Protocol $\Pi$ securely computes $\mathcal{F}$ if for each $i \in \{0, 1\}$, there exists a probabilistic polynomial-time (PPT) simulator $\mathsf{Sim}_i$ such that:
\begin{equation}
    \left\{ \mathsf{Sim}_i(1^\lambda, \mathsf{inp}_i, \mathcal{F}_i(\mathsf{inp}_0, \mathsf{inp}_1)) \right\} \stackrel{c}{\equiv} \left\{ \mathsf{View}^{\Pi}_i(1^\lambda, \mathsf{inp}_0, \mathsf{inp}_1) \right\}
\end{equation}
where $\mathsf{View}^{\Pi}_i$ denotes the view of $P_i$ (input, randomness, received messages), $\mathcal{F}_i$ denotes the output to $P_i$, and $\stackrel{c}{\equiv}$ denotes computational indistinguishability.
\end{definition}

\section{Methodology}
\label{sec:SecureCROWN}

\subsection{Overview}
\label{subsec:overview}

\textbf{Challenges.}
Implementing the robustness verification (Section~\ref{subsec:nn_verification}) in MPC presents two fundamental challenges.
\emph{(1) Data-dependent branching:} The relaxation slope $\alpha$ depends on neuron activation state—inactive ($\bar{z} \leq 0$), active ($\underline{z} \geq 0$), or unstable—each requiring distinct computation (Eq.~\eqref{eq:relu_relax}). Similarly, intercept terms $\underline{\boldsymbol{\delta}}$ and $\bar{\boldsymbol{\delta}}$ (Eqs.~\eqref{eq:c_lb_update}--\eqref{eq:c_ub_update}) depend on coefficient signs. In MPC, such data-dependent branching requires costly secure comparison, preventing efficient batching across neurons.
\emph{(2) Precision sensitivity in deep
recursion:} The backward recursion in Eq.~\eqref{eq:A_update} updates the coefficient matrix $\mathbf{A}^{(t)}$ across layers, with relaxation slopes $\boldsymbol{\alpha}^{(t)}$ computed via division (Eq.~\eqref{eq:relu_relax}). In fixed-point MPC, division has limited precision, and truncation errors accumulate through the propagation chain, potentially compromising verification accuracy in deep networks.

\textbf{Our Approach.}
We address these challenges through two key techniques.
\emph{(1) Branch Elimination via ReLU:} We observe that all branching 
conditions in the verifier are fundamentally sign tests. Since ReLU 
inherently encodes sign information, we reformulate conditional logic 
into unified arithmetic expressions where ReLU outputs naturally select 
the correct terms. This eliminates branching entirely, enabling identical 
computation paths for all neurons and facilitating efficient vectorized 
execution in MPC. These ReLU operations are efficiently realized via 
FSS-based protocols with constant-round communication.
\emph{(2) Newton--Raphson Refinement:} To address precision loss in secure division, we apply Newton--Raphson iterations to refine reciprocal estimates, approximately doubling the precision bits per iteration at minimal additional cost.

\textbf{Protocol Overview.}
The protocol operates in the preprocessing model: correlated randomness is generated offline based on the public network architecture, reducing the online phase to local computations and share exchanges.
Specifically, in the \textbf{offline phase}, the trusted dealer $\mathcal{D}$ generates correlated randomness based on the network architecture $\{d_0, \ldots, d_L\}$, including Beaver triples for multiplication, FSS keys with associated input masks for function evaluations, and other correlated randomness required by the sub-protocols. This phase is data-independent, so $\mathcal{D}$ learns nothing about private inputs.
The \textbf{online phase} begins with \textbf{input sharing}, where the parties secret-share their private inputs as described in Section~\ref{sec:problem}. 
Next, in the \textbf{secure computation} phase, the parties jointly execute $\Pi_{\mathrm{Verify}}$ (Algorithm~\ref{alg:verify}): for each hidden layer $l \in \{1, \ldots, L-1\}$, compute pre-activation bounds $\langle \underline{\mathbf{z}}^{(l)} \rangle, \langle \bar{\mathbf{z}}^{(l)} \rangle$ via $\Pi_{\mathrm{Backward}}$ (Algorithm~\ref{alg:lbp}) and derive relaxation slopes $\langle \boldsymbol{\alpha}^{(l)} \rangle$ via $\Pi_{\boldsymbol{\alpha}}$; then execute a final backward pass using $\langle \mathbf{d} \rangle$ to compute $\langle \underline{f}_{\tcls,j} \rangle$ (Eq.~\eqref{eq:final_lb}).
Finally, in the \textbf{output} phase, $P_0$ sends $\langle \underline{f}_{\tcls,j} \rangle_0$ to $P_1$, who reconstructs the certified lower bound $\underline{f}_{\tcls,j}$ and concludes local robustness if $\underline{f}_{\tcls,j} > 0$.

\subsection{Secure Building Blocks}
\label{subsec:building_blocks}

\textbf{Secure Multiplication ($\mathsf{SecMul}$, $\mathsf{SecMatMul}$).}
Secure scalar multiplication $\mathsf{SecMul}(\langle x \rangle, \langle y \rangle) \rightarrow \langle xy \rangle$ follows the standard Beaver triple protocol (Section~\ref{subsec:2pc}). For matrix multiplication, given inputs $\mathbf{X} \in \mathbb{Z}_{2^k}^{m \times n}$ and $\mathbf{Y} \in \mathbb{Z}_{2^k}^{n \times p}$, the protocol computes $\mathsf{SecMatMul}(\langle \mathbf{X} \rangle, \langle \mathbf{Y} \rangle) \rightarrow \langle \mathbf{XY} \rangle$ by consuming matrix-level Beaver triples $(\langle \mathbf{A} \rangle, \langle \mathbf{B} \rangle, \langle \mathbf{AB} \rangle)$, thereby avoiding the overhead of scalar-wise decomposition.

\textbf{Secure Arithmetic Right Shift ($\mathsf{SecARS}$).}
To realign the radix point after fixed-point multiplications—which inherently yield $2n_f$ fractional bits—we must restore the standard $n_f$-bit precision. We achieve this via $\mathsf{SecARS}(\langle x \rangle, n_f) \rightarrow \langle \lfloor x / 2^{n_f} \rfloor \rangle$, efficiently instantiated using the DCF-based protocol of~\cite{Gupta2025shark}.

\textbf{Secure ReLU and Absolute Value ($\mathsf{SecReLU}$, $\mathsf{SecAbs}$).}
Both primitives are built on securely computing the sign indicator $\mu = \mathbb{I}[x < 0]$ via DCF. Parties evaluate preprocessed DCF keys to obtain the secret-shared bit $\langle \mu \rangle$, and subsequently compute $\mathsf{SecReLU}(\langle x \rangle) \rightarrow \mathsf{SecMul}(\langle 1 - \mu \rangle, \langle x \rangle)$ and $\mathsf{SecAbs}(\langle x \rangle) \rightarrow \mathsf{SecMul}(\langle 1 - 2\mu \rangle, \langle x \rangle)$.
Each primitive strictly requires one DCF evaluation followed by one $\mathsf{SecMul}$ operation.

\textbf{Secure Reciprocal ($\mathsf{SecRecip}$).}
We implement the reciprocal protocol $\mathsf{SecRecip}(\langle x \rangle) \rightarrow \langle 1/x \rangle$ in two stages. First, we adapt the FSS-based approach from~\cite{Gupta2025shark} to generate an initial approximation. The input magnitude is decomposed as $|x| \approx (1 + \frac{m}{2^{\ell_m}}) \cdot 2^e$, and a seed $w_0 \approx 1/|x|$ is computed via $w_0 = 2^{-e} \cdot \frac{2^{\ell_m}}{2^{\ell_m} + m}$, where the terms are retrieved via spline and lookup tables. Second, since the lookup-based initialization provides limited precision, we perform $n_{\mathrm{iter}}$ Newton--Raphson iterations $w_{i+1} = w_i \cdot (2 - |x| \cdot w_i)$ to refine the result to the target precision. Each iteration roughly doubles the number of accurate bits, ensuring sufficient numerical stability for subsequent verification steps.

\subsection{Secure Slope Computation \texorpdfstring{$\Pi_{\bm{\alpha}}$}{Pi-alpha}}
\label{subsec:alpha}

The relaxation slope vector $\boldsymbol{\alpha}^{(l)}$ in Eq.~\eqref{eq:relu_relax} is essential for constructing the linear propagation matrices in LBP. Although the definition is simple, securely evaluating the piecewise logic in Eq.~\eqref{eq:relu_relax} is costly in MPC, as it requires secure comparisons and data-dependent branching per neuron to distinguish active, inactive, and unstable states. To circumvent these overheads and enable efficient execution, we reformulate the piecewise logic into a unified, data-independent arithmetic expression applied element-wise across all neurons in the layer:
\begin{equation}\label{eq:alpha_unified}
\boldsymbol{\alpha} = \frac{\text{ReLU}(\bar{\mathbf{z}})}{\text{ReLU}(\bar{\mathbf{z}}) + \text{ReLU}(-\underline{\mathbf{z}}) + \epsilon_s}.
\end{equation}
This formulation naturally encapsulates the three theoretical cases within a single data-independent execution path: for inactive neurons where $\bar{z} \leq 0$, the vanishing numerator $\text{ReLU}(\bar{z}) = 0$ yields a slope of $\alpha = 0$; for active neurons where $\underline{z} \geq 0$, the term $\text{ReLU}(-\underline{z}) = 0$ results in $\alpha \approx 1$; and for unstable neurons, the expression simplifies to the optimal chord slope $\bar{z}/(\bar{z} - \underline{z})$. 
To ensure numerical robustness during the MPC execution, we introduce a small public stability constant $\epsilon_s$ into the denominator of Eq.~\eqref{eq:alpha_unified}. This safeguard prevents potential division-by-zero errors that may arise when pre-activation bounds $[\underline{z}, \overline{z}]$ vanish near zero---a common phenomenon in deep neural networks due to the accumulation of conservative bounds during propagation.
{A theoretical analysis of its impact on the certified bounds is provided in Appendix~\ref{appendix:error-analysis}.}
The detailed vectorized protocol $\Pi_{\boldsymbol{\alpha}}$, which achieves branching-free computation using $\mathsf{SecReLU}$, $\mathsf{SecRecip}$, and $\mathsf{SecMul}$ primitives, is presented in Algorithm~\ref{alg:alpha} in Appendix~\ref{app:algorithms}.

\subsection{Secure Intercept Accumulation \texorpdfstring{$\Pi_{\boldsymbol{\delta}}$}{Pi-delta}}

\label{subsec:bias}

As established in Eqs.~\eqref{eq:c_lb_update}--\eqref{eq:c_ub_update}, the intercept contribution vectors $\underline{\boldsymbol{\delta}}^{(t-1)}, \bar{\boldsymbol{\delta}}^{(t-1)} \in \mathbb{R}^{d_l}$ depend on the element-wise signs of the intermediate coefficient matrix $\hat{\mathbf{A}}^{(t)} \in \mathbb{R}^{d_l \times d_{t-1}}$. From Eq.~\eqref{eq:relu_relax}, the lower relaxation intercept is $\underline{\beta}^{(t-1)}_j = 0$ for all neurons, while the upper intercept is $\bar{\beta}^{(t-1)}_j = -\alpha^{(t-1)}_j \underline{z}^{(t-1)}_j$ for unstable neurons (and zero otherwise).
Since $\underline{\beta}^{(t-1)}_j = 0$, the accumulations simplify significantly: only terms involving $\bar{\beta}$ contribute. By substituting $A^{(t-1)}_{ij} = \hat{A}^{(t)}_{ij} \alpha^{(t-1)}_j$ (Eq.~\eqref{eq:A_update}) and noting that $\alpha^{(t-1)}_j \geq 0$ preserves signs, the $i$-th elements of the lower and upper contributions reduce to $\underline{\delta}^{(t-1)}_i = \sum_{j:\, A^{(t-1)}_{ij} \leq 0} A^{(t-1)}_{ij} (-\underline{z}^{(t-1)}_j)$ and $\bar{\delta}^{(t-1)}_i = \sum_{j:\, A^{(t-1)}_{ij} > 0} A^{(t-1)}_{ij} (-\underline{z}^{(t-1)}_j)$, respectively.

To avoid expensive secure comparisons in MPC, these conditional summations can be elegantly reformulated into unified matrix-vector multiplications using ReLU operations:
\begin{align}
\underline{\boldsymbol{\delta}}^{(t-1)} &= -\text{ReLU}(-\mathbf{A}^{(t-1)}) \text{ReLU}(-\underline{\mathbf{z}}^{(t-1)}), \label{eq:delta_lower_unified} \\
\bar{\boldsymbol{\delta}}^{(t-1)} &= \text{ReLU}(\mathbf{A}^{(t-1)}) \text{ReLU}(-\underline{\mathbf{z}}^{(t-1)}), \label{eq:delta_upper_unified}
\end{align}
The protocol $\Pi_{\boldsymbol{\delta}}$ is given in Algorithm~\ref{alg:delta} (Appendix~\ref{app:algorithms}).

\begin{algorithm}[t]
\algtext*{EndFor}
\caption{$\Pi_{\mathrm{Backward}}$: Secure Linear Bound Propagation}
\label{alg:lbp}
\begin{algorithmic}[1]
\Require Target layer $l$; $\{\langle \boldsymbol{\alpha}^{(t)} \rangle, \langle \underline{\mathbf{z}}^{(t)} \rangle\}_{t=1}^{l-1}$; $\langle \mathbf{x}_0 \rangle$; $\langle \epsilon \rangle$
\Ensure Concrete bounds $\langle \underline{\mathbf{z}}^{(l)} \rangle, \langle \bar{\mathbf{z}}^{(l)} \rangle$
\State $\langle \mathbf{A} \rangle \gets \langle \mathbf{I}_{d_l} \rangle$; \quad $\langle \underline{\mathbf{c}} \rangle \gets \langle \mathbf{0} \rangle$; \quad $\langle \bar{\mathbf{c}} \rangle \gets \langle \mathbf{0} \rangle$
\For{$t = l, \ldots, 2$}
    \State $\langle \hat{\mathbf{A}} \rangle \gets \mathsf{SecMatMul}(\langle \mathbf{A} \rangle, \langle \mathbf{W}^{(t)} \rangle)$ \hfill $\triangleright$ Eq.~\eqref{eq:Ahat_update}
    \State $\langle \mathbf{A} \rangle \gets \mathsf{SecMatMul}(\langle \hat{\mathbf{A}} \rangle, \langle \operatorname{diag}({\boldsymbol{\alpha}}^{(t-1)}) \rangle)$
    \Statex \hfill $\triangleright$ Eq.~\eqref{eq:A_update}
    \State $\langle \underline{\boldsymbol{\delta}} \rangle, \langle \bar{\boldsymbol{\delta}} \rangle \gets \Pi_{\boldsymbol{\delta}}(\langle \mathbf{A} \rangle, \langle \underline{\mathbf{z}}^{(t-1)} \rangle)$ \hfill $\triangleright$ Eqs.~\eqref{eq:delta_lower_unified}--\eqref{eq:delta_upper_unified}
    \State $\langle \underline{\mathbf{c}} \rangle \gets \langle \underline{\mathbf{c}} \rangle + \mathsf{SecMatMul}(\langle \mathbf{A} \rangle, \langle \mathbf{b}^{(t)} \rangle) + \langle \underline{\boldsymbol{\delta}} \rangle$
    \Statex \hfill $\triangleright$ Eq.~\eqref{eq:c_lb_update}
    \State $\langle \bar{\mathbf{c}} \rangle \gets \langle \bar{\mathbf{c}} \rangle + \mathsf{SecMatMul}(\langle \mathbf{A} \rangle, \langle \mathbf{b}^{(t)} \rangle) + \langle \bar{\boldsymbol{\delta}} \rangle$
    \Statex \hfill $\triangleright$ Eq.~\eqref{eq:c_ub_update}
\EndFor
\State $\langle \mathbf{A}^{(0)} \rangle \gets \mathsf{SecMatMul}(\langle \mathbf{A} \rangle, \langle \mathbf{W}^{(1)} \rangle)$
\State $\langle \underline{\mathbf{c}}^{(0)} \rangle \gets \langle \underline{\mathbf{c}} \rangle + \mathsf{SecMatMul}(\langle \mathbf{A} \rangle, \langle \mathbf{b}^{(1)} \rangle)$
\State $\langle \bar{\mathbf{c}}^{(0)} \rangle \gets \langle \bar{\mathbf{c}} \rangle + \mathsf{SecMatMul}(\langle \mathbf{A} \rangle, \langle \mathbf{b}^{(1)} \rangle)$
\State $\langle \mathbf{r} \rangle \gets \mathrm{RowSum}(\mathsf{SecAbs}(\langle \mathbf{A}^{(0)} \rangle))$
\Statex \hfill $\triangleright$ Row-wise $\ell_1$-norm
\State $\langle \mathbf{m} \rangle \gets \mathsf{SecMatMul}(\langle \mathbf{A}^{(0)} \rangle, \langle \mathbf{x}_0 \rangle)$;
\Statex  $\langle \mathbf{s} \rangle \gets \mathsf{SecMul}(\langle \epsilon \rangle, \langle \mathbf{r} \rangle)$
\State $\langle \underline{\mathbf{z}}^{(l)} \rangle \gets \langle \mathbf{m} \rangle + \langle \underline{\mathbf{c}}^{(0)} \rangle - \langle \mathbf{s} \rangle$;
\Statex  $\langle \bar{\mathbf{z}}^{(l)} \rangle \gets \langle \mathbf{m} \rangle + \langle \bar{\mathbf{c}}^{(0)} \rangle + \langle \mathbf{s} \rangle$
\State \Return $\langle \underline{\mathbf{z}}^{(l)} \rangle, \langle \bar{\mathbf{z}}^{(l)} \rangle$
\end{algorithmic}
\end{algorithm}

\begin{algorithm}[!htp]
\algtext*{EndFor}
\caption{$\Pi_{\mathrm{Verify}}$: Secure Robustness Verification}
\label{alg:verify}
\begin{algorithmic}[1]
\Require Shares $\langle \mathbf{W}^{(l)} \rangle$, $\langle \mathbf{b}^{(l)} \rangle$ for $l \in [L]$; $\langle \mathbf{x}_0 \rangle$; $\langle \epsilon \rangle$; $\langle \mathbf{d} \rangle$
\Ensure Certified lower bound $\langle \underline{f}_{y,j} \rangle$
\Statex \textbf{/* Iterative Bound and Slope Computation */}
\State $\langle \mathbf{z}^{(1)} \rangle \gets \mathsf{SecMatMul}(\langle \mathbf{W}^{(1)} \rangle, \langle \mathbf{x}_0 \rangle) + \langle \mathbf{b}^{(1)} \rangle$
\State $\langle \mathbf{r}^{(1)} \rangle \gets \mathrm{RowSum}(\mathsf{SecAbs}(\langle \mathbf{W}^{(1)} \rangle))$ \hfill $\triangleright$ Row-wise $\ell_1$-norm
\State $\langle \bar{\mathbf{z}}^{(1)} \rangle \gets \langle \mathbf{z}^{(1)} \rangle + \mathsf{SecMul}(\langle \epsilon \rangle, \langle \mathbf{r}^{(1)} \rangle)$
\State $\langle \underline{\mathbf{z}}^{(1)} \rangle \gets \langle \mathbf{z}^{(1)} \rangle - \mathsf{SecMul}(\langle \epsilon \rangle, \langle \mathbf{r}^{(1)} \rangle)$
\State $\langle \boldsymbol{\alpha}^{(1)} \rangle \gets \Pi_{\alpha}(\langle \bar{\mathbf{z}}^{(1)} \rangle, \langle \underline{\mathbf{z}}^{(1)} \rangle)$
\For{$l = 2, \ldots, L-1$}
    \State $\langle \underline{\mathbf{z}}^{(l)} \rangle, \langle \bar{\mathbf{z}}^{(l)} \rangle \gets \Pi_{\mathrm{Backward}}(l, \{\langle \boldsymbol{\alpha}^{(t)} \rangle, \langle \underline{\mathbf{z}}^{(t)} \rangle\}_{t=1}^{l-1})$
    \State $\langle \boldsymbol{\alpha}^{(l)} \rangle \gets \Pi_{\alpha}(\langle \bar{\mathbf{z}}^{(l)} \rangle, \langle \underline{\mathbf{z}}^{(l)} \rangle)$
\EndFor
\Statex \textbf{/* Margin Certification */}
\State $\langle \mathbf{W}^{(L)} \rangle \gets \mathsf{SecMatMul}(\langle \mathbf{d} \rangle^\top, \langle \mathbf{W}^{(L)} \rangle)$
\State $\langle b^{(L)} \rangle \gets \mathsf{SecMatMul}(\langle \mathbf{d} \rangle^\top, \langle \mathbf{b}^{(L)} \rangle)$
\State $\langle \underline{f}_{y,j} \rangle, \_ \gets \Pi_{\mathrm{Backward}}(L, \{\langle \boldsymbol{\alpha}^{(t)} \rangle, \langle \underline{\mathbf{z}}^{(t)} \rangle\}_{t=1}^{L-1})$
\State \Return $\langle \underline{f}_{y,j} \rangle$
\end{algorithmic}
\end{algorithm}

\subsection{Secure Verification Protocol}
\label{subsec:protocol}
Building on the sub-protocols $\Pi_{\boldsymbol{\alpha}}$ and $\Pi_{\boldsymbol{\delta}}$, we now present the complete secure verification protocol. Algorithm~\ref{alg:lbp} ($\Pi_{\mathrm{Backward}}$) securely implements the backward recursion (Eqs.~\eqref{eq:Ahat_update}--\eqref{eq:c_ub_update}) for a target layer $l$. 
Starting with $\langle \mathbf{A} \rangle = \langle \mathbf{I}_{d_l} \rangle$ and $\langle \underline{\mathbf{c}} \rangle = \langle \bar{\mathbf{c}} \rangle = \langle \mathbf{0} \rangle$, the protocol iterates for $t = l, \ldots, 2$, performing coefficient propagation ($\mathsf{SecMatMul}$), slope multiplication ($\mathsf{SecMul}$), intercept accumulation ($\Pi_{\boldsymbol{\delta}}$), and bias addition. 
Final bounds follow from H\"older's inequality (Eqs.~\eqref{eq:z_lb}--\eqref{eq:z_ub}), with the $\ell_1$-norm computed via $\mathsf{SecAbs}$ and local summation (for $\ell_\infty$-perturbations).

\textbf{Iterative Bound and Slope Computation.}
As described in Section~\ref{subsec:nn_verification}, computing bounds at layer $l$ requires slopes from all preceding layers, so the parties compute bounds and slopes iteratively from $l = 1$ to $L-1$. For $l = 1$, bounds are computed directly: the linear term $\langle \mathbf{W}^{(1)} \mathbf{x}_0 \rangle$ via $\mathsf{SecMatMul}$ and the row-wise norm via $\mathsf{SecAbs}$; slopes $\langle \boldsymbol{\alpha}^{(1)} \rangle$ then follow from $\Pi_{\boldsymbol{\alpha}}$. For $l \geq 2$, the parties invoke $\Pi_{\mathrm{Backward}}$ with previously computed slope shares to obtain $\langle \underline{\mathbf{z}}^{(l)} \rangle, \langle \bar{\mathbf{z}}^{(l)} \rangle$, and then derive $\langle \boldsymbol{\alpha}^{(l)} \rangle$ via $\Pi_{\boldsymbol{\alpha}}$.

\textbf{Margin Certification.}
After obtaining all slope shares $\langle \boldsymbol{\alpha}^{(1)} \rangle, \ldots, \langle \boldsymbol{\alpha}^{(L-1)} \rangle$, the parties compute the certified margin $\langle \underline{f}_{\tcls,j} \rangle$ via a final backward pass. To avoid $O(K)$ overhead from secure array indexing, $P_1$ provides a difference vector $\langle \mathbf{d} \rangle$ where $\mathbf{d} = \mathbf{e}_\tcls - \mathbf{e}_j \in \{-1, 0, 1\}^K$. The backward pass initializes $\langle \mathbf{W}^{(L)} \rangle \leftarrow \mathsf{SecMatMul}(\langle \mathbf{d} \rangle^\top, \langle \mathbf{W}^{(L)} \rangle) \in \mathbb{R}^{1 \times d_{L-1}}$ and $\langle b^{(L)} \rangle \leftarrow \mathsf{SecMul}(\langle \mathbf{d} \rangle^\top, \langle \mathbf{b}^{(L)} \rangle) \in \mathbb{R}$. This reduces the coefficient $\mathbf{A}^{(t)}$ to a row vector in $\mathbb{R}^{1 \times d_t}$, lowering per-layer communication from $O(d^2)$ to $O(d)$. The parties then invoke $\Pi_{\mathrm{Backward}}$ with target layer $L$ to obtain $\langle \mathbf{A}^{(0)} \rangle$ and $\langle \underline{c}^{(0)} \rangle$, from which the certified margin $\langle \underline{f}_{\tcls,j} \rangle$ follows via Eq.~\eqref{eq:final_lb}.

\section{Theoretical Analysis}
\label{sec:analysis}

\textbf{Security.}
The security of \name{} is analyzed under the semi-honest adversary model (Definition~\ref{def:semihonest_security}). The proof constructs a simulator via a hybrid argument over $L$ layers, reducing security to the underlying FSS and ASS primitives.
The complete proof of Theorem~\ref{thm:security} is provided in Appendix~\ref{app:security_proof}.
{A pathway for upgrading to malicious security, together with an empirical cost estimate, is discussed in Appendix~\ref{app:malicious}.}
\begin{theorem}[Security]\label{thm:security}
Protocol $\Pi_{\mathrm{verify}}$ securely computes $\mathcal{F}_{\mathrm{verify}}$ in the $\mathcal{F}_{\mathcal{D}}$-hybrid model against semi-honest adversaries.
\end{theorem}
\textbf{Complexity.}
For an $L$-layer network with maximum width $d$, \name{} requires $O(L^2)$ rounds, $O(L^2 d^2 k)$ communication, and $O(L^2 d^3)$ computation (Appendix~\ref{app:complexity}).
\textbf{Error Analysis.}
Fixed-point arithmetic introduces error in $\lb{f}_{\tcls,j}$ bounded by a network-depth-dependent  multiple of the per-operation rounding unit $\varepsilon_q = 2^{-n_f}$(Appendix~\ref{appendix:error-analysis}).
With maximum weight norm $\rho = \max_i\|\mathbf{W}^{(i)}\|_\infty$, 
the absolute error in each output bound entry scales as $O(L)\,\varepsilon_q$, $O(L^2)\,\varepsilon_q$, and $O(\rho^L)\,\varepsilon_q$ for the stable ($\rho<1$), unit-norm ($\rho=1$), and unstable ($\rho>1$) regimes, respectively.

\section{Experimental Evaluation}\label{sec:experiments}

\subsection{Experimental Setup}
\label{subsec:Exsetup}

\textbf{Implementation.}
We implement \name{} in C++, with $k=64$-bit fixed-point representation, $n_f=26$ fractional bits, stability constant $\epsilon_s=10^{-4}$, and $n_{\mathrm{iter}}=1$ Newton--Raphson refinement iteration in $\mathsf{SecRecip}$ (the FSS-based seed provides approximately $\lceil n_f/2 \rceil$ accurate bits, and one iteration doubles this to ${\geq}\, n_f$ bits).
Experiments are run on a machine with an Intel Xeon w7-3445 CPU and 64\,GB of RAM, using four CPU threads.
We simulate network conditions via Linux \texttt{tc}: LAN (10 Gbps bandwidth, 0.05ms round-trip time (RTT)) and WAN (370/600\,Mbps, 40/60\,ms RTT). Results are averaged over 10 trials.
Our replication package is available at \href{https://github.com/songzhibo1/SecureCROWN}{https://github.com/songzhibo1/SecureCROWN}.

\begin{table}[t]
\centering
\setlength{\tabcolsep}{2.5pt}
\renewcommand{\arraystretch}{1.05}
\caption{Fidelity analysis of \name{}. $m \times [n]$: $m$ layers, $n$ neurons each. R: Robust ($\underline{f}_{\tcls,j} > 0$); U: Unknown ($\underline{f}_{\tcls,j} \leq 0$).}
\label{tab:fidelity_final_complete}
\resizebox{\columnwidth}{!}{%
\begin{tabular}{ccccccc}
\toprule
& & \textbf{Avg.} & & & & \\
\textbf{Data} & \textbf{Model} & \textbf{Rob.} & \textbf{$\epsilon$} & \textbf{Result} & \textbf{MRE} & \textbf{Cons.} \\
\midrule
\multirow{21}{*}{\rotatebox{90}{MNIST}} 
& \multirow{4}{*}{\rotatebox{90}{$2\times[20]$}} & \multirow{4}{*}{\rotatebox{90}{0.0448}} 
& $0.015$ & \cellcolor{gray!15}\textbf{All R (99)} & \cellcolor{gray!15}$\mathbf{4.45 \times 10^{-7}}$ & $\mathbf{100\%}$ \\
\cmidrule{4-7}
& & & \multirow{2}{*}{$0.045$} & \cellcolor{gray!15}\textbf{R (43)} & \cellcolor{gray!15}$\mathbf{9.52 \times 10^{-7}}$ & $\mathbf{100\%}$ \\
& & & & U (56) & $8.70 \times 10^{-7}$ & $\mathbf{100\%}$ \\
\cmidrule{4-7}
& & & $0.1$ & All U (99) & $3.27 \times 10^{-7}$ & $\mathbf{100\%}$ \\ 
\cmidrule{2-7}
& \multirow{5}{*}{\rotatebox{90}{$3\times[20]$}} & \multirow{5}{*}{\rotatebox{90}{0.0304}} 
& \multirow{2}{*}{$0.015$} & \cellcolor{gray!15}\textbf{R (93)} & \cellcolor{gray!15}$\mathbf{1.41 \times 10^{-6}}$ & $\mathbf{100\%}$ \\
& & & & U (1) & $1.08 \times 10^{-4}$ & $\mathbf{100\%}$ \\
\cmidrule{4-7}
& & & \multirow{2}{*}{$0.030$} & \cellcolor{gray!15}\textbf{R (44)} & \cellcolor{gray!15}$\mathbf{7.18 \times 10^{-6}}$ & $\mathbf{100\%}$ \\ 
& & & & U (50) & $3.68 \times 10^{-6}$ & $\mathbf{100\%}$ \\
\cmidrule{4-7}
& & & $0.1$ & All U (94) & $6.85 \times 10^{-7}$ & $\mathbf{100\%}$ \\
\cmidrule{2-7}
& \multirow{3}{*}{\rotatebox{90}{$3\times[256]$}} & \multirow{3}{*}{\rotatebox{90}{0.0159}} 
& \multirow{2}{*}{$0.015$} & \cellcolor{gray!15}\textbf{R (49)} & \cellcolor{gray!15}$\mathbf{6.20 \times 10^{-6}}$ & $\mathbf{100\%}$ \\
& & & & U (51) & $1.14 \times 10^{-5}$ & $\mathbf{100\%}$ \\
\cmidrule{4-7}
& & & $0.1$ & All U (100) & $9.29 \times 10^{-7}$ & $\mathbf{100\%}$ \\
\cmidrule{2-7}
& \multirow{3}{*}{\rotatebox{90}{$5\times[256]$}} & \multirow{3}{*}{\rotatebox{90}{0.0149}} 
& \multirow{2}{*}{$0.015$} & \cellcolor{gray!15}\textbf{R (50)} & \cellcolor{gray!15}$\mathbf{7.78 \times 10^{-6}}$ & $\mathbf{100\%}$ \\
& & & & U (50) & $1.21 \times 10^{-5}$ & $\mathbf{100\%}$ \\
\cmidrule{4-7}
& & & $0.1$ & All U (100) & $6.16 \times 10^{-5}$ & $\mathbf{100\%}$ \\
\cmidrule{2-7}
& \multirow{3}{*}{\rotatebox{90}{$7\times[256]$}} & \multirow{3}{*}{\rotatebox{90}{0.0132}} 
& \multirow{2}{*}{$0.015$} & \cellcolor{gray!15}\textbf{R (30)} & \cellcolor{gray!15}$\mathbf{2.66 \times 10^{-5}}$ & $\mathbf{100\%}$ \\
& & & & U (69) & $1.32 \times 10^{-2}$ & $\mathbf{100\%}$ \\
\cmidrule{4-7}
& & & $0.1$ & All U (99) & $9.18 \times 10^{-2}$ & $\mathbf{100\%}$ \\
\midrule
\multirow{7}{*}{\rotatebox{90}{CIFAR-10}} 
& \multirow{3}{*}{\rotatebox{90}{$5\times[100]$}} & \multirow{3}{*}{\rotatebox{90}{0.0021}} 
& \multirow{2}{*}{$0.002$} & \cellcolor{gray!15}\textbf{R (24)} & \cellcolor{gray!15}$\mathbf{5.23 \times 10^{-5}}$ & $\mathbf{100\%}$ \\
& & & & U (21) & $7.78 \times 10^{-5}$ & $\mathbf{100\%}$ \\
\cmidrule{4-7}
& & & $0.0078$ & All U (45) & $5.65 \times 10^{-5}$ & $\mathbf{100\%}$ \\
\cmidrule{2-7}
& \multirow{2}{*}[0.5em]{\rotatebox{90}{$7{\times}[100]$}} & \multirow{2}{*}[0em]{\rotatebox{90}{0.0013}} 
& \multirow{2}{*}[0.1em]{$0.001$} & \cellcolor{gray!15}\rule{0pt}{1.3em}\textbf{R (32)} & \cellcolor{gray!15}$\mathbf{2.82 \times 10^{-4}}$ & $\mathbf{100\%}$ \\
& & & & \rule[-1.2em]{0pt}{0.2em}U (14) & $1.50 \times 10^{-4}$ & $\mathbf{100\%}$ \\
\cmidrule{2-7}
& \multirow{2}{*}[0.5em]{\rotatebox{90}{$10{\times}[200]$}} & \multirow{2}{*}[0em]{\rotatebox{90}{0.0010}}
& \multirow{2}{*}[-0.1em]{$0.001$} & \cellcolor{gray!15}\rule{0pt}{1.3em}\textbf{R (28)} & \cellcolor{gray!15}$\mathbf{1.73 \times 10^{-3}}$ & $\mathbf{100\%}$ \\
& & & & \rule[-1.2em]{0pt}{0.2em}U (24) & $1.67 \times 10^{-4}$ & $\mathbf{100\%}$ \\
\bottomrule
\end{tabular}
}
\end{table}

\begin{figure*}[!htbp]
    \centering
\includegraphics[width=\linewidth]{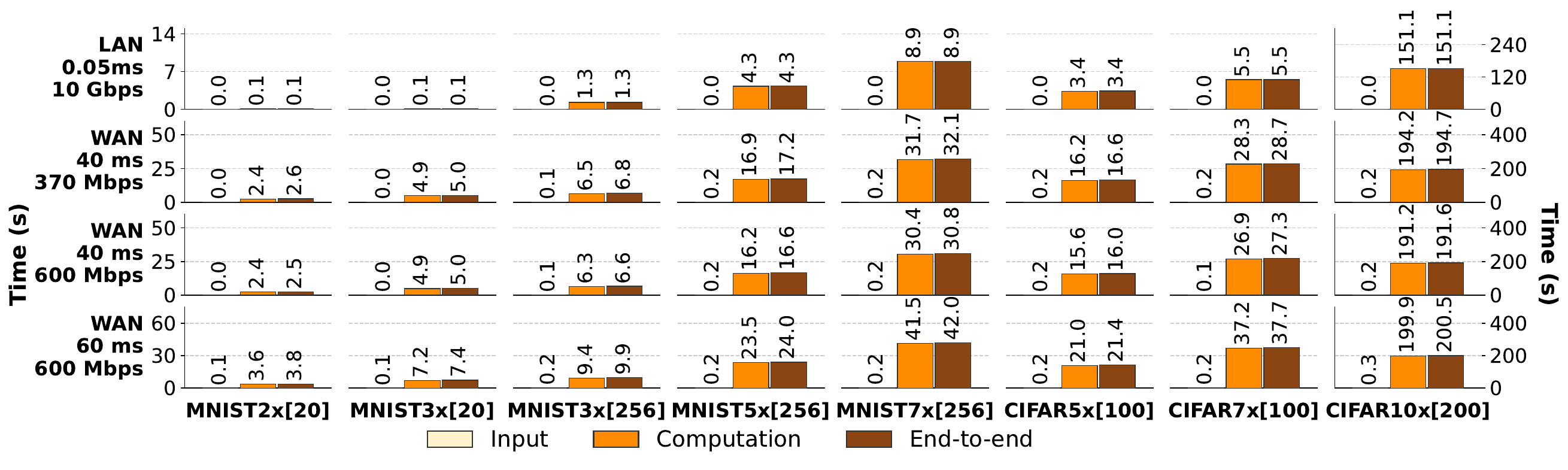}
\caption{{Online latency breakdown across network conditions and model architectures}.}
    \label{fig:time_breakdown}
\end{figure*}

\textbf{Baseline.}
As no prior work addresses privacy-preserving robustness verification, we compare against the plaintext CROWN algorithm (Section~\ref{subsec:nn_verification}), implemented in Python and executed single-threaded. This serves as the correctness oracle for validating numerical fidelity.

\textbf{Models and Datasets.}
We evaluate on MNIST~\citep{lecun1998gradient} and CIFAR-10~\citep{krizhevsky2009learning} using fully-connected ReLU networks. 
MNIST models ($2 \times [20]$ to $7 \times [256]$) are {from~\citep{zhang2018efficient} and VNN-COMP~\citep{brix2023vnncomp}};
CIFAR-10 models ($5 \times [100]$ to $10 \times [200]$) are from ERAN~\citep{singh2019abstract}. For each configuration, we randomly sample 100 test images, verify those that are correctly classified, and test multiple $\epsilon$ values under $\ell_\infty$ perturbations—including values near the robustness boundary (mixed Robust/Unknown outcomes) and extreme values (all Robust or all Unknown).

\subsection{Results}
\label{subsec:results}

\textbf{RQ1: Fidelity.}
We validate \name{}'s correctness by measuring numerical fidelity against the plaintext baseline using two metrics.
The first is \emph{Mean Relative Error} (MRE):
$\text{MRE} = \frac{1}{N} \sum_{i=1}^{N} | (\underline{f}_{\tcls,j}^{\text{sec},(i)} - \underline{f}_{\tcls,j}^{\text{plain},(i)}) / \underline{f}_{\tcls,j}^{\text{plain},(i)} |$,
where $\underline{f}_{\tcls,j}^{\text{sec}}$ and $\underline{f}_{\tcls,j}^{\text{plain}}$ denote the \name{} and plaintext certified lower bounds, respectively. The second is \emph{Verification Consistency}, the percentage of samples with identical verification outcomes.

As shown in Table~\ref{tab:fidelity_final_complete}, \name{} achieves high numerical fidelity, with MRE scaling predictably with network depth. Our evaluation specifically examines challenging cases near the robustness boundary, as reflected by the mixed Robust/Unknown outcomes in Table~\ref{tab:fidelity_final_complete}. For shallow networks (MNIST $2 \times [20]$), MRE remains at $\approx 10^{-7}$. As depth increases, cumulative rounding errors in recursive bound propagation lead to higher MRE---from $6.20 \times 10^{-6}$ (MNIST $3 \times [256]$) to $1.73 \times 10^{-3}$ for CIFAR-10 $10 \times [200]$. Despite this, \name{} maintains \textbf{100\% verification consistency} with the plaintext baseline across all 8 architectures and multiple $\epsilon$ configurations, confirming that our approach faithfully preserves the decision behavior of the plaintext verifier.

\begin{table}[tbp]
\centering
\caption{Online communication volume breakdown (MB).}
\label{tab:comm_breakdown}
\setlength{\tabcolsep}{5pt} %
\begin{tabular}{cc|SSS} 
\toprule
\multirow{2}{*}{\textbf{Dataset}} & \multirow{2}{*}{\textbf{Model}} & \multicolumn{3}{c}{\textbf{Communication (MB)}} \\ 

\cmidrule(lr){3-5}
 & & {\textbf{Input}} & {\textbf{Comp.}} & {\textbf{E2E}} \\
\midrule

\multirow{5}{*}{MNIST}
 & $2 \times [20]$ & 0.13 & 0.64 & 0.77 \\
 & $3 \times [20]$ & 0.14 & 1.78 & 1.92 \\
 & $3 \times [256]$ & 2.07 & 30.16 & 32.23 \\
 & $5 \times [256]$ & 3.07 & 109.95 & 113.02 \\
 & $7 \times [256]$ & 4.07 & 234.80 & 238.88 \\

\midrule

\multirow{3}{*}{CIFAR-10}
 & $5 \times [100]$ & 2.63 & 83.76 & 86.39 \\
 & $7 \times [100]$ & 2.78 & 140.92 & 143.71 \\
 & $10 \times [200]$ & 7.20 & 596.97 & 604.17 \\

\bottomrule
\end{tabular}
\end{table}

\textbf{RQ2: Online Efficiency.}
We evaluate online efficiency by analyzing runtime latency and communication volume across network conditions and model architectures. We break down the end-to-end (E2E) execution into two phases: (1) \emph{Input Processing}, encompassing secure input sharing and constraint setup; (2) \emph{Computation}, the secure backward bound propagation with FSS-based nonlinear primitives. Figure~\ref{fig:time_breakdown} and Table~\ref{tab:comm_breakdown} summarize the results.

\textbf{Latency.} 
As shown in Figure~\ref{fig:time_breakdown}, E2E latency ranges from 0.1s to 200.5s across network conditions and model architectures.
Smaller models (MNIST $2\times[20]$ to $3\times[256]$) complete within 0.1--9.9s, while CIFAR-10 $10\times[200]$ requires 151.1--200.5s. The computation phase dominates total latency, with input processing contributing only 0.0--0.3s. Bandwidth variation at fixed RTT (370 vs 600\,Mbps, 40ms) yields $<$2\% difference, whereas increasing RTT from 40ms to 60ms incurs $\sim$5\% overhead. This indicates that round-trip latency accumulates over $O(L^2)$ communication rounds, making RTT the dominant network factor.
{FSS evaluations are embarrassingly parallel across neurons, so exploiting more cores would further reduce latency for larger models.}

\textbf{Communication.}
Table~\ref{tab:comm_breakdown} details communication volume. E2E communication scales linearly with network complexity: from 0.77 MB (MNIST $2 \times [20]$) to 604.17 MB (CIFAR-10 $10 \times [200]$). The computation phase dominates, accounting for $\approx$98.8\% of total volume in the largest model. Input processing remains negligible ($\leq$7.20 MB even for CIFAR-10), confirming that \emph{communication overhead scales with secure operations rather than input dimensions}, ensuring scalability to higher-dimensional data.

\textbf{{RQ3: Offline Preprocessing Overhead.}}
As shown in Table~\ref{tab:preprocessing_stats}, both time and storage scale with network complexity, ranging from 0.17\,s / 39\,MB (MNIST $2 \times [20]$) to 71.16\,s / 31\,GB (CIFAR-10 $10 \times [200]$). While storage costs are substantial for deep networks, the preprocessing phase is \emph{one-time and data-independent}---generated offline without access to private inputs. 
{In practical deployment, preprocessing material can be prepared as a batch for multiple planned verification queries; each query consumes a fresh portion of the batch during execution, amortizing the generation cost across those queries.}


\begin{table}[tbp]
\centering
\caption{{Offline preprocessing overhead (one-time, data-independent).}}
\label{tab:preprocessing_stats}
\setlength{\tabcolsep}{3.5pt}
\renewcommand{\arraystretch}{1}
\begin{tabular}{llSS}
\toprule
\textbf{Dataset} & \textbf{Model} & {\textbf{Time (s)}} & {\textbf{Data (MB)}} \\
\midrule
\multirow{5}{*}{MNIST}
  & $2 \times [20]$ & 0.17 & 39 \\
  & $3 \times [20]$ & 0.41 & 100 \\
  & $3 \times [256]$ & 4.20 & 1644 \\
  & $5 \times [256]$ & 13.93 & 5677 \\
  & $7 \times [256]$ & 29.21 & 11898 \\
\midrule
\multirow{3}{*}{CIFAR-10}
  & $5 \times [100]$ & 12.03 & 4547 \\
  & $7 \times [100]$ & 19.25 & 7541 \\
  & $10 \times [200]$ & 71.16 & 31060 \\
\bottomrule
\end{tabular}
\end{table}

\textbf{{RQ4: Sensitivity to the Stability Constant.}}
\label{sec:rq4}
{We vary the stability constant $\epsilon_s$ 
(Eq.~(\ref{eq:alpha_unified})) to examine its impact on 
verification outcomes. Specifically, $\epsilon_s$ is varied from 
$10^{-4}$ (our default) to $10^{-1}$ on MNIST $5{\times}[256]$ 
($\epsilon = 0.015$). 
As shown in Table~\ref{tab:ablation_eps}, verification consistency 
remains 100\% across all tested values. At the same time, the reported 
MREs increase approximately linearly with $\epsilon_s$, indicating that 
larger stability constants introduce larger numerical deviations from 
the plaintext bounds. On this benchmark, verification 
decisions remain unchanged even when $\epsilon_s$ is increased to 
$1000\times$ the default value.}

\begin{table}[tbp]
\centering
\caption{{Sensitivity to the stability constant $\epsilon_s$ on MNIST $5{\times}[256]$ 
($\epsilon = 0.015$).}}
\label{tab:ablation_eps}
\begin{tabular}{cccc}
\toprule
\multirow{2}{*}{$\boldsymbol{\epsilon_s}$} & 
\multirow{2}{*}{\textbf{Cons.}} & 
\multicolumn{2}{c}{\textbf{Avg.\ MRE}} \\
\cmidrule(lr){3-4}
 & & \textbf{R (50)} & \textbf{U (50)} \\
\midrule
$10^{-1}$ & 100\% & $5.86 \times 10^{-2}$ & $1.40 \times 10^{-1}$ \\
$10^{-2}$ & 100\% & $6.47 \times 10^{-3}$ & $1.50 \times 10^{-2}$ \\
$10^{-3}$ & 100\% & $6.49 \times 10^{-4}$ & $1.52 \times 10^{-3}$ \\
$10^{-4}$ & 100\% & $7.05 \times 10^{-5}$ & $1.62 \times 10^{-4}$ \\
\bottomrule
\end{tabular}
\end{table}

\section{Conclusion}
We propose \name{}, the first privacy-preserving framework for neural network robustness verification in a 2PC setting.
Our protocol addresses the key challenges of performing linear relaxation and backward bound computation on protected model parameters and inputs.
We provided formal privacy guarantees in the semi-honest adversary model and analyzed the protocol's computational complexity. 
Experiments on standard benchmarks validate that our method achieves verification fidelity comparable to plaintext verification with moderate cryptographic overhead, representing a promising step toward efficient privacy-preserving verification.

\begin{acknowledgements}
This work was supported by the National Science and Technology Major Project (Grant 2022ZD0115901), the University of Bristol Engineering and Physical Sciences Research Council Impact Acceleration Account, and the National Natural Science Foundation of China (Grant Nos. 92582102, 62572013, and 62172019). Nianyun Song is supported by the China Scholarship Council (Grant 202506040106). We thank the anonymous reviewers for their constructive feedback.
\end{acknowledgements}

\bibliography{ref}

\begin{thebibliography}{42}
\providecommand{\natexlab}[1]{#1}
\providecommand{\url}[1]{\texttt{#1}}
\expandafter\ifx\csname urlstyle\endcsname\relax
  \providecommand{\doi}[1]{doi: #1}\else
  \providecommand{\doi}{doi: \begingroup \urlstyle{rm}\Url}\fi

\bibitem[Althoff(2015)]{Althoff15CORA}
Matthias Althoff.
\newblock An introduction to {CORA} 2015.
\newblock In \emph{1st and 2nd International Workshop on Applied Verification for Continuous and Hybrid Systems}, volume~34 of \emph{EPiC Series in Computing}, pages 120--151. EasyChair, 2015.

\bibitem[Boyle et~al.(2015)Boyle, Gilboa, and Ishai]{BoyleGI15}
Elette Boyle, Niv Gilboa, and Yuval Ishai.
\newblock Function secret sharing.
\newblock In \emph{Advances in Cryptology, {EUROCRYPT} 2015}, volume 9057 of \emph{Lecture Notes in Computer Science}, pages 337--367. Springer, 2015.

\bibitem[Brix et~al.(2023)Brix, M{\"{u}}ller, Bak, Johnson, and Liu]{brix2023vnncomp}
Christopher Brix, Mark~Niklas M{\"{u}}ller, Stanley Bak, Taylor~T. Johnson, and Changliu Liu.
\newblock First three years of the international verification of neural networks competition {(VNN-COMP)}.
\newblock \emph{Int. J. Softw. Tools Technol. Transf.}, 25\penalty0 (3):\penalty0 329--339, 2023.

\bibitem[Davenport and Kalakota(2019)]{davenport2019potential}
Thomas Davenport and Ravi Kalakota.
\newblock The potential for artificial intelligence in healthcare.
\newblock \emph{Future Healthcare Journal}, 6\penalty0 (2):\penalty0 94--98, 2019.

\bibitem[Duong et~al.(2023)Duong, Li, Nguyen, and Dwyer]{Hai2023neuralsat}
Hai Duong, Linhan Li, ThanhVu Nguyen, and Matthew~B. Dwyer.
\newblock A {DPLL(T)} framework for verifying deep neural networks.
\newblock \emph{CoRR}, abs/2307.10266, 2023.

\bibitem[Duong et~al.(2024)Duong, Xu, Nguyen, and Dwyer]{Duong2024neuralsat}
Hai Duong, Dong Xu, ThanhVu Nguyen, and Matthew~B. Dwyer.
\newblock Harnessing neuron stability to improve {DNN} verification.
\newblock \emph{Proc. {ACM} Softw. Eng.}, 1\penalty0 ({FSE}):\penalty0 859--881, 2024.

\bibitem[Duong et~al.(2025)Duong, Nguyen, and Dwyer]{Duong2025NeuralSAT}
Hai Duong, ThanhVu Nguyen, and Matthew~B. Dwyer.
\newblock {NeuralSAT}: {A} high-performance verification tool for deep neural networks.
\newblock In \emph{Computer Aided Verification, {CAV} 2025}, volume 15932 of \emph{Lecture Notes in Computer Science}, pages 409--423. Springer, 2025.

\bibitem[Esteva et~al.(2017)Esteva, Kuprel, Novoa, Ko, Swetter, Blau, and Thrun]{esteva2017dermatologist}
Andre Esteva, Brett Kuprel, Roberto~A Novoa, Justin Ko, Susan~M Swetter, Helen~M Blau, and Sebastian Thrun.
\newblock Dermatologist-level classification of skin cancer with deep neural networks.
\newblock \emph{Nature}, 542\penalty0 (7639):\penalty0 115--118, 2017.

\bibitem[Esteva et~al.(2019)Esteva, Robicquet, Ramsundar, Kuleshov, DePristo, Chou, Cui, Corrado, Thrun, and Dean]{esteva2019guide}
Andre Esteva, Alexandre Robicquet, Bharath Ramsundar, Volodymyr Kuleshov, Mark DePristo, Katherine Chou, Claire Cui, Greg Corrado, Sebastian Thrun, and Jeff Dean.
\newblock A guide to deep learning in healthcare.
\newblock \emph{Nature Medicine}, 25\penalty0 (1):\penalty0 24--29, 2019.

\bibitem[Evans et~al.(2018)Evans, Kolesnikov, and Rosulek]{Evans2018MPC}
David Evans, Vladimir Kolesnikov, and Mike Rosulek.
\newblock A pragmatic introduction to secure multi-party computation.
\newblock \emph{Found. Trends Priv. Secur.}, 2\penalty0 (2-3):\penalty0 70--246, 2018.

\bibitem[Goldreich et~al.(1987)Goldreich, Micali, and Wigderson]{goldreich1987how}
O.~Goldreich, S.~Micali, and A.~Wigderson.
\newblock How to play {ANY} mental game.
\newblock In \emph{Proceedings of the Annual {ACM} Symposium on Theory of Computing, {STOC} 1987}, pages 218--229. {ACM}, 1987.

\bibitem[Goodfellow et~al.(2015)Goodfellow, Shlens, and Szegedy]{Goodfellow2014explain}
Ian~J. Goodfellow, Jonathon Shlens, and Christian Szegedy.
\newblock Explaining and harnessing adversarial examples.
\newblock In \emph{International Conference on Learning Representations, {ICLR} 2015}, 2015.

\bibitem[Gupta et~al.(2025)Gupta, Chandran, Gupta, Katz, and Sharma]{Gupta2025shark}
Kanav Gupta, Nishanth Chandran, Divya Gupta, Jonathan Katz, and Rahul Sharma.
\newblock {SHARK:} actively secure inference using function secret sharing.
\newblock In \emph{{IEEE} Symposium on Security and Privacy, {SP} 2025}, pages 2472--2490. {IEEE}, 2025.

\bibitem[Hendrycks and Dietterich(2019)]{Hendrycks2019Benchmark}
Dan Hendrycks and Thomas~G. Dietterich.
\newblock Benchmarking neural network robustness to common corruptions and perturbations.
\newblock In \emph{International Conference on Learning Representations, {ICLR} 2019}, 2019.

\bibitem[Katz et~al.(2017)Katz, Barrett, Dill, Julian, and Kochenderfer]{katz2017reluplex}
Guy Katz, Clark Barrett, David~L. Dill, Kyle Julian, and Mykel~J. Kochenderfer.
\newblock Reluplex: An efficient {SMT} solver for verifying deep neural networks.
\newblock In \emph{Computer Aided Verification, {CAV} 2017}, volume 10426 of \emph{Lecture Notes in Computer Science}, pages 97--117. Springer, 2017.

\bibitem[Katz et~al.(2019)Katz, Huang, Ibeling, Julian, Lazarus, Lim, Shah, Thakoor, Wu, Zeljic, Dill, Kochenderfer, and Barrett]{Katz2019Marabou}
Guy Katz, Derek~A. Huang, Duligur Ibeling, Kyle Julian, Christopher Lazarus, Rachel Lim, Parth Shah, Shantanu Thakoor, Haoze Wu, Aleksandar Zeljic, David~L. Dill, Mykel~J. Kochenderfer, and Clark~W. Barrett.
\newblock The {Marabou} framework for verification and analysis of deep neural networks.
\newblock In \emph{Computer Aided Verification, {CAV} 2019}, volume 11561 of \emph{Lecture Notes in Computer Science}, pages 443--452. Springer, 2019.

\bibitem[Knott et~al.(2021)Knott, Venkataraman, Hannun, Sengupta, Ibrahim, and van~der Maaten]{Knott2021CrypTen}
Brian Knott, Shobha Venkataraman, Awni~Y. Hannun, Shubho Sengupta, Mark Ibrahim, and Laurens van~der Maaten.
\newblock {CrypTen}: Secure multi-party computation meets machine learning.
\newblock \emph{Advances in Neural Information Processing Systems}, 34:\penalty0 4961--4973, 2021.

\bibitem[Krizhevsky(2009)]{krizhevsky2009learning}
Alex Krizhevsky.
\newblock Learning multiple layers of features from tiny images.
\newblock Technical Report TR-2009, Department of Computer Science, University of Toronto, 2009.

\bibitem[Kumar et~al.(2020)Kumar, Rathee, Chandran, Gupta, Rastogi, and Sharma]{Kumar2020cryptflow}
Nishant Kumar, Mayank Rathee, Nishanth Chandran, Divya Gupta, Aseem Rastogi, and Rahul Sharma.
\newblock {CrypTFlow}: Secure {TensorFlow} inference.
\newblock In \emph{{IEEE} Symposium on Security and Privacy, {SP} 2020}, pages 336--353. {IEEE}, 2020.

\bibitem[LeCun et~al.(1998)LeCun, Bottou, Bengio, and Haffner]{lecun1998gradient}
Yann LeCun, L{\'e}on Bottou, Yoshua Bengio, and Patrick Haffner.
\newblock Gradient-based learning applied to document recognition.
\newblock \emph{Proceedings of the IEEE}, 86\penalty0 (11):\penalty0 2278--2324, 1998.

\bibitem[Lee et~al.(2022)Lee, Lee, Lee, Kim, Kim, No, and Choi]{Lee2022HE}
Eunsang Lee, Joon{-}Woo Lee, Junghyun Lee, Young{-}Sik Kim, Yongjune Kim, Jong{-}Seon No, and Woosuk Choi.
\newblock Low-complexity deep convolutional neural networks on fully homomorphic encryption using multiplexed parallel convolutions.
\newblock In \emph{International Conference on Machine Learning, {ICML} 2022}, volume 162 of \emph{Proceedings of Machine Learning Research}, pages 12403--12422. {PMLR}, 2022.

\bibitem[Lee et~al.(2024)Lee, Ko, Kim, and Oh]{Lee2024zk}
Seunghwa Lee, Hankyung Ko, Jihye Kim, and Hyunok Oh.
\newblock {vCNN}: Verifiable convolutional neural network based on zk-{SNARKs}.
\newblock \emph{{IEEE} Trans. Dependable Secur. Comput.}, 21\penalty0 (4):\penalty0 4254--4270, 2024.

\bibitem[Litjens et~al.(2017)Litjens, Kooi, Bejnordi, Setio, Ciompi, Ghafoorian, Van Der~Laak, Van~Ginneken, and S{\'a}nchez]{litjens2017survey}
Geert Litjens, Thijs Kooi, Babak~Ehteshami Bejnordi, Arnaud Arindra~Adiyoso Setio, Francesco Ciompi, Mohsen Ghafoorian, Jeroen~Awm Van Der~Laak, Bram Van~Ginneken, and Clara~I S{\'a}nchez.
\newblock A survey on deep learning in medical image analysis.
\newblock \emph{Medical Image Analysis}, 42:\penalty0 60--88, 2017.

\bibitem[Lou and Jiang(2021)]{Lou2021HEMET}
Qian Lou and Lei Jiang.
\newblock {HEMET:} {A} homomorphic-encryption-friendly privacy-preserving mobile neural network architecture.
\newblock In \emph{International Conference on Machine Learning, {ICML} 2021}, volume 139 of \emph{Proceedings of Machine Learning Research}, pages 7102--7110. {PMLR}, 2021.

\bibitem[Maheri et~al.(2025)Maheri, Haddadi, and Davidson]{Maheri2025telesparse}
Mohammad~Mahdi Maheri, Hamed Haddadi, and Alex Davidson.
\newblock {TeleSparse}: Practical privacy-preserving verification of deep neural networks.
\newblock \emph{Proc. Priv. Enhancing Technol.}, 2025\penalty0 (4):\penalty0 861--880, 2025.

\bibitem[Mohassel and Rindal(2018)]{mohassel2018aby3}
Payman Mohassel and Peter Rindal.
\newblock {ABY3}: A mixed protocol framework for machine learning.
\newblock In \emph{Proceedings of the {ACM} {SIGSAC} Conference on Computer and Communications Security, {CCS} 2018}, pages 35--52. {ACM}, 2018.

\bibitem[Salman et~al.(2019)Salman, Yang, Zhang, Hsieh, and Zhang]{salman2019convex}
Hadi Salman, Greg Yang, Huan Zhang, Cho-Jui Hsieh, and Pengchuan Zhang.
\newblock A convex relaxation barrier to tight robustness verification of neural networks.
\newblock \emph{Advances in Neural Information Processing Systems}, 32:\penalty0 9835--9846, 2019.

\bibitem[Shi et~al.(2025)Shi, Jin, Kolter, Jana, Hsieh, and Zhang]{shi2025genbab}
Zhouxing Shi, Qirui Jin, Zico Kolter, Suman Jana, Cho-Jui Hsieh, and Huan Zhang.
\newblock Neural network verification with branch-and-bound for general nonlinearities.
\newblock In \emph{Tools and Algorithms for the Construction and Analysis of Systems, {TACAS} 2025}, volume 15696 of \emph{Lecture Notes in Computer Science}, pages 315--335. Springer, 2025.

\bibitem[Singh et~al.(2019)Singh, Gehr, P{\"u}schel, and Vechev]{singh2019abstract}
Gagandeep Singh, Timon Gehr, Markus P{\"u}schel, and Martin Vechev.
\newblock An abstract domain for certifying neural networks.
\newblock \emph{Proc. {ACM} Program. Lang.}, 3\penalty0 ({POPL}):\penalty0 1--30, 2019.

\bibitem[Song et~al.(2026)Song, Wang, Bie, Guo, Zhang, and Jia]{obligate}
Nianyun Song, Fuyi Wang, Rongfang Bie, Yu~Guo, Leo~Yu Zhang, and Xiaohua Jia.
\newblock {ObliviGate}: Towards architecture-oblivious privacy-preserving inference for malicious security.
\newblock \emph{IEEE Transactions on Services Computing}, 2026.

\bibitem[Szegedy et~al.(2014)Szegedy, Zaremba, Sutskever, Bruna, Erhan, Goodfellow, and Fergus]{Szegedy13intriguing}
Christian Szegedy, Wojciech Zaremba, Ilya Sutskever, Joan Bruna, Dumitru Erhan, Ian~J. Goodfellow, and Rob Fergus.
\newblock Intriguing properties of neural networks.
\newblock In \emph{International Conference on Learning Representations, {ICLR} 2014}, 2014.

\bibitem[Tan et~al.(2021)Tan, Knott, Tian, and Wu]{Tan2021cryptgpu}
Sijun Tan, Brian Knott, Yuan Tian, and David~J. Wu.
\newblock {CryptGPU}: Fast privacy-preserving machine learning on the {GPU}.
\newblock In \emph{{IEEE} Symposium on Security and Privacy, {SP} 2021}, pages 1021--1038. {IEEE}, 2021.

\bibitem[Voigt and Von~dem Bussche(2017)]{voigt2017eu}
Paul Voigt and Axel Von~dem Bussche.
\newblock \emph{The {EU} General Data Protection Regulation ({GDPR}): A Practical Guide}.
\newblock Springer International Publishing, 2017.

\bibitem[Wagh et~al.(2019)Wagh, Gupta, and Chandran]{wagh2019securenn}
Sameer Wagh, Divya Gupta, and Nishanth Chandran.
\newblock {SecureNN}: 3-party secure computation for neural network training.
\newblock \emph{Proc. Priv. Enhancing Technol.}, 2019\penalty0 (3):\penalty0 26--49, 2019.

\bibitem[Wagh et~al.(2021)Wagh, Tople, Benhamouda, Kushilevitz, Mittal, and Rabin]{Wagh2021falcon}
Sameer Wagh, Shruti Tople, Fabrice Benhamouda, Eyal Kushilevitz, Prateek Mittal, and Tal Rabin.
\newblock Falcon: Honest-majority maliciously secure framework for private deep learning.
\newblock \emph{Proc. Priv. Enhancing Technol.}, 2021\penalty0 (1):\penalty0 188--208, 2021.

\bibitem[Wang et~al.(2021)Wang, Zhang, Xu, Lin, Jana, Hsieh, and Kolter]{wang2021beta}
Shiqi Wang, Huan Zhang, Kaidi Xu, Xue Lin, Suman Jana, Cho-Jui Hsieh, and J~Zico Kolter.
\newblock {Beta-CROWN}: Efficient bound propagation with per-neuron split constraints for complete and incomplete neural network verification.
\newblock \emph{Advances in Neural Information Processing Systems}, 34:\penalty0 29909--29921, 2021.

\bibitem[Wu et~al.(2024)Wu, Isac, Zeljic, Tagomori, Daggitt, Kokke, Refaeli, Amir, Julian, Bassan, Huang, Lahav, Wu, Zhang, Komendantskaya, Katz, and Barrett]{Wu2024Marabou}
Haoze Wu, Omri Isac, Aleksandar Zeljic, Teruhiro Tagomori, Matthew~L. Daggitt, Wen Kokke, Idan Refaeli, Guy Amir, Kyle Julian, Shahaf Bassan, Pei Huang, Ori Lahav, Min Wu, Min Zhang, Ekaterina Komendantskaya, Guy Katz, and Clark~W. Barrett.
\newblock Marabou 2.0: {A} versatile formal analyzer of neural networks.
\newblock In \emph{Computer Aided Verification, {CAV} 2024}, volume 14682 of \emph{Lecture Notes in Computer Science}, pages 249--264. Springer, 2024.

\bibitem[Xu et~al.(2020)Xu, Shi, Zhang, Wang, Chang, Huang, Kailkhura, Lin, and Hsieh]{xu2020automatic}
Kaidi Xu, Zhouxing Shi, Huan Zhang, Yihan Wang, Kai-Wei Chang, Minlie Huang, Bhavya Kailkhura, Xue Lin, and Cho-Jui Hsieh.
\newblock Automatic perturbation analysis for scalable certified robustness and beyond.
\newblock \emph{Advances in Neural Information Processing Systems}, 33:\penalty0 1129--1141, 2020.

\bibitem[Xu et~al.(2021)Xu, Zhang, Wang, Wang, Jana, Lin, and Hsieh]{xu2021fast}
Kaidi Xu, Huan Zhang, Shiqi Wang, Yihan Wang, Suman Jana, Xue Lin, and Cho-Jui Hsieh.
\newblock {Fast and Complete}: Enabling complete neural network verification with rapid and massively parallel incomplete verifiers.
\newblock In \emph{International Conference on Learning Representations, {ICLR} 2021}, 2021.

\bibitem[Yao(1982)]{yao1982protocols}
Andrew~C Yao.
\newblock Protocols for secure computations.
\newblock In \emph{Proceedings of the Annual Symposium on Foundations of Computer Science, {FOCS} 1982}, pages 160--164. {IEEE}, 1982.

\bibitem[Zhang et~al.(2018)Zhang, Weng, Chen, Hsieh, and Daniel]{zhang2018efficient}
Huan Zhang, Tsui-Wei Weng, Pin-Yu Chen, Cho-Jui Hsieh, and Luca Daniel.
\newblock Efficient neural network robustness certification with general activation functions.
\newblock \emph{Advances in Neural Information Processing Systems}, 31:\penalty0 4939--4948, 2018.

\bibitem[Zhang et~al.(2021)Zhang, Xin, and Wu]{Zhang2021Survey}
Qiao Zhang, Chunsheng Xin, and Hongyi Wu.
\newblock Privacy-preserving deep learning based on multiparty secure computation: A survey.
\newblock \emph{{IEEE} Internet Things J.}, 8\penalty0 (13):\penalty0 10412--10429, 2021.

\end{thebibliography}

\newpage
\onecolumn
\title{Appendix}

\makeatletter
\renewcommand{\@thanks}{}
\makeatother

\setcounter{footnote}{0}
\maketitle
\appendix
\renewcommand{\thetheorem}{\Alph{section}.\arabic{theorem}}
\renewcommand{\thelemma}{\Alph{section}.\arabic{lemma}}

\setcounter{theorem}{0}
\setcounter{lemma}{0}

\section{Detailed Protocols}
\label{app:algorithms}

\subsection{Secure Relaxation Slope Computation}
Algorithm~\ref{alg:alpha} details the secure protocol $\Pi_{\boldsymbol{\alpha}}$ for computing the relaxation slope vector $\boldsymbol{\alpha}^{(l)}$, directly implementing the unified arithmetic expression derived in Eq.~\eqref{eq:alpha_unified} of the main text. As discussed in Section~\ref{subsec:alpha}, securely evaluating the piecewise conditions of the standard ReLU relaxation introduces prohibitive communication overhead due to data-dependent branching and secure comparisons. By formulating the slope computation as a continuous, data-independent function, our protocol evaluates all neurons uniformly. The protocol first extracts the required positive bounds and the magnitudes of the negative bounds using standard $\mathsf{SecReLU}$ operations. After locally incorporating the public stability constant $\epsilon_s$ to prevent numerical instability, the division is securely resolved by computing the reciprocal of the denominator via $\mathsf{SecRecip}$, followed by an element-wise $\mathsf{SecMul}$. This branching-free design ensures deterministic execution time and minimizes the use of heavy cryptographic primitives.

\begin{algorithm}[ht]
\caption{$\Pi_{\boldsymbol{\alpha}}$: Secure Slope Computation}
\label{alg:alpha}
\begin{algorithmic}[1]
\Require Secret shares of pre-activation lower bounds $\langle \underline{\mathbf{z}}^{(l)} \rangle \in \mathbb{Z}_{2^k}^{d_l}$, upper bounds $\langle \bar{\mathbf{z}}^{(l)} \rangle \in \mathbb{Z}_{2^k}^{d_l}$, and public stability constant $\epsilon_s \in \mathbb{R}^{+}$
\Ensure Secret shares of the relaxation slope vector $\langle \boldsymbol{\alpha}^{(l)} \rangle \in \mathbb{Z}_{2^k}^{d_l}$

\State $\langle \mathbf{n} \rangle \gets \mathsf{SecReLU}(\langle \bar{\mathbf{z}}^{(l)} \rangle)$ \Comment{Numerator: $\mathrm{ReLU}(\bar{\mathbf{z}}^{(l)})$}
\State $\langle \mathbf{d}_{sub} \rangle \gets \mathsf{SecReLU}(-\langle \underline{\mathbf{z}}^{(l)} \rangle)$ \Comment{Partial denominator: $\mathrm{ReLU}(-\underline{\mathbf{z}}^{(l)})$}

\State $\langle \mathbf{d} \rangle \gets \langle \mathbf{n} \rangle + \langle \mathbf{d}_{sub} \rangle + \epsilon_s$ \Comment{Full denominator via local addition (zero communication)}

\State $\langle \mathbf{r} \rangle \gets \mathsf{SecRecip}(\langle \mathbf{d} \rangle)$ \Comment{Element-wise secure reciprocal: $1 / \mathbf{d}$}
\State $\langle \boldsymbol{\alpha}^{(l)} \rangle \gets \mathsf{SecMul}(\langle \mathbf{n} \rangle, \langle \mathbf{r} \rangle)$ \Comment{Element-wise secure multiplication for division}

\State \Return $\langle \boldsymbol{\alpha}^{(l)} \rangle$
\end{algorithmic}
\end{algorithm}

\subsection{Secure Intercept Accumulation}
Algorithm~\ref{alg:delta} details the secure protocol $\Pi_{\boldsymbol{\delta}}$ for computing the intercept contribution vectors $\underline{\boldsymbol{\delta}}^{(t-1)}$ and $\bar{\boldsymbol{\delta}}^{(t-1)}$, implementing the reformulated expressions in Eqs.~\eqref{eq:delta_lower_unified}--\eqref{eq:delta_upper_unified}. As discussed in Section~\ref{subsec:bias}, directly evaluating the conditional summations would necessitate expensive secure comparison protocols. By leveraging the property that the relaxation slopes satisfy $\boldsymbol{\alpha}^{(t-1)} \ge 0$, the updated coefficients $\mathbf{A}^{(t-1)}$ preserve the necessary sign information. The protocol efficiently isolates the positive and negative components of the coefficient matrix $\langle \mathbf{A}^{(t-1)} \rangle$ and the negative components of the pre-activation bounds $\langle \underline{\mathbf{z}}^{(t-1)} \rangle$ using three parallel $\mathsf{SecReLU}$ evaluations. The final vectors $\langle \underline{\boldsymbol{\delta}}^{(t-1)} \rangle$ and $\langle \bar{\boldsymbol{\delta}}^{(t-1)} \rangle$ are then obtained via $\mathsf{SecMatMul}$.

\begin{algorithm}[ht]
\caption{$\Pi_{\boldsymbol{\delta}}$: Secure Intercept Accumulation}
\label{alg:delta}
\begin{algorithmic}[1]
\Require Secret shares of coefficients $\langle \mathbf{A}^{(t-1)} \rangle \in \mathbb{Z}_{2^k}^{d_l \times d_{t-1}}$ and pre-activation lower bounds $\langle \underline{\mathbf{z}}^{(t-1)} \rangle \in \mathbb{Z}_{2^k}^{d_{t-1}}$
\Ensure Secret shares of intercept contributions $\langle \underline{\boldsymbol{\delta}}^{(t-1)} \rangle, \langle \bar{\boldsymbol{\delta}}^{(t-1)} \rangle \in \mathbb{Z}_{2^k}^{d_l}$

\State $\langle \mathbf{R}_A^{-} \rangle \gets \mathsf{SecReLU}(-\langle \mathbf{A}^{(t-1)} \rangle)$ \Comment{$\operatorname{ReLU}(-\mathbf{A}^{(t-1)})$, element-wise}
\State $\langle \mathbf{R}_A^{+} \rangle \gets \mathsf{SecReLU}(\langle \mathbf{A}^{(t-1)} \rangle)$ \Comment{$\operatorname{ReLU}(\mathbf{A}^{(t-1)})$, element-wise}
\State $\langle \mathbf{r}_z \rangle \gets \mathsf{SecReLU}(-\langle \underline{\mathbf{z}}^{(t-1)} \rangle)$ \Comment{$\operatorname{ReLU}(-\underline{\mathbf{z}}^{(t-1)})$}

\State $\langle \underline{\boldsymbol{\delta}}^{(t-1)} \rangle \gets -\mathsf{SecMatMul}(\langle \mathbf{R}_A^{-} \rangle, \langle \mathbf{r}_z \rangle)$ \Comment{Eq.~\eqref{eq:delta_lower_unified}}
\State $\langle \bar{\boldsymbol{\delta}}^{(t-1)} \rangle \gets \mathsf{SecMatMul}(\langle \mathbf{R}_A^{+} \rangle, \langle \mathbf{r}_z \rangle)$ \Comment{Eq.~\eqref{eq:delta_upper_unified}}

\State \Return $\langle \underline{\boldsymbol{\delta}}^{(t-1)} \rangle, \langle \bar{\boldsymbol{\delta}}^{(t-1)} \rangle$
\end{algorithmic}
\end{algorithm}

\section{Detailed Analysis}
\label{app:analysis}

\subsection{Security Proof}
\label{app:security_proof}

\begin{proof}
We prove security via a hybrid argument. Let the verification circuit consist of $L$ layers corresponding to the neural network's bound propagation. For each layer $l$, let $\mathsf{F}_l$ denote the FSS scheme used (either for $\mathsf{SecReLU}$, $\mathsf{SecMul}$, or $\mathsf{SecRecip}$), with associated simulator $\mathsf{Sim}_l$.

We construct simulator $\mathsf{Sim}_0$ for corrupted model owner $P_0$; the case for data owner $P_1$ is symmetric.

\paragraph{Hybrid Definitions.}
\begin{itemize}
    \item $\mathbf{H}_0$: The simulator knows all inputs and executes the real protocol honestly. This is identical to the real execution.
    \item $\mathbf{H}_l$ (for $l \in \{1, \ldots, L\}$): Execute the real protocol for layers $1, \ldots, L-l$. For layers $L-l+1, \ldots, L$, use the FSS simulator $\mathsf{Sim}_l$ to generate keys and simulate protocol messages.
    \item $\mathbf{H}_{L+1}$: Same as $\mathbf{H}_L$, but additionally replace $P_1$'s inputs with dummy values. This is the ideal-world execution.
\end{itemize}

\paragraph{Indistinguishability of Adjacent Hybrids.}

\emph{$\mathbf{H}_l \approx_c \mathbf{H}_{l+1}$ for $l \in \{0, \ldots, L-1\}$:}
These hybrids differ only in layer $j = L - l$. Suppose a distinguisher $\mathcal{D}$ distinguishes $\mathbf{H}_l$ from $\mathbf{H}_{l+1}$ with non-negligible advantage. We construct an adversary $\mathcal{A}$ that breaks the security of FSS scheme $\mathsf{F}_j$:
\begin{enumerate}
    \item $\mathcal{A}$ runs the protocol for layers $1, \ldots, j-1$ honestly.
    \item $\mathcal{A}$ receives an FSS key $k_b$ from the FSS challenger (either real or simulated).
    \item $\mathcal{A}$ uses $k_b$ to execute layer $j$ and simulates layers $j+1, \ldots, L$.
    \item $\mathcal{A}$ outputs $\mathcal{D}$'s guess.
\end{enumerate}
If $\mathcal{D}$ distinguishes, then $\mathcal{A}$ breaks FSS security—a contradiction.

\emph{$\mathbf{H}_L \approx_s \mathbf{H}_{L+1}$:}
In both hybrids, $P_1$'s inputs are masked by uniform random shares. Since additive secret sharing is information-theoretically secure, replacing real inputs with dummy inputs produces an identical distribution of shares. Thus these hybrids are \emph{statistically} indistinguishable.

\paragraph{Conclusion.}
By the hybrid argument, $\mathbf{H}_0 \approx_c \mathbf{H}_{L+1}$, which establishes that the real and ideal executions are computationally indistinguishable.
\end{proof}

\subsection{Complexity Analysis}
\label{app:complexity}

We analyze the theoretical complexity of \name{} for an $L$-layer network with maximum width $d$. Table~\ref{tab:primitives} summarizes the costs of the underlying building blocks.

\begin{table}[ht]
\centering
\setlength{\tabcolsep}{4pt}
\caption{Complexity of secure building blocks. \textbf{Rounds} and \textbf{Comm.} denote online costs; \textbf{Preproc.} denotes offline costs. Parameters: bit-width $k = 64$, fractional bits $n_f = 26$, security parameter $\lambda = 128$, Newton--Raphson iterations $n_{\mathrm{iter}} = 1$.}
\label{tab:primitives}
\begin{tabular}{lccc}
\toprule
\textbf{Primitive} & \textbf{Rounds} & \textbf{Comm. (bits)} & \textbf{Preproc. (bits)} \\
\midrule
SecMul & 1 & $2k$ \; (128) & $3k$ \; (192) \\
SecMatMul$_{m \times n \times p}$ & 1 & $2(mn{+}np)k$ & $(mn{+}np{+}mp)k$ \\
SecARS & 2 & $k + 2$ \; (66) & $(k{+}n_f)\lambda + O(k)$ \\
SecReLU & 2 & $k + 1$ \; (65) & $k\lambda + O(k)$ \\
SecAbs & 2 & $k + 1$ \; (65) & $k\lambda + O(k)$ \\
SecRecip & 6 & $O(k)$ & $O(k\lambda)$ \\
\bottomrule
\end{tabular}
\end{table}

\paragraph{Sub-protocol Complexity.}
The complexity of our specialized sub-protocols is derived as follows:
\begin{itemize}
    \item \textbf{Protocol $\Pi_{\boldsymbol{\alpha}}$ (Algorithm~\ref{alg:alpha}):} This involves two parallel $\mathsf{SecReLU}$ operations (2 rounds), one $\mathsf{SecRecip}$ ($4 + 2n_{\mathrm{iter}}$ rounds, including $n_{\mathrm{iter}}$ Newton--Raphson iterations), and a final multiplication (1 round). The total is $7 + 2n_{\mathrm{iter}}$ rounds, with communication $O(dk)$.
    \item \textbf{Protocol $\Pi_{\boldsymbol{\delta}}$ (Algorithm~\ref{alg:delta}):} This protocol executes three parallel $\mathsf{SecReLU}$ evaluations (2 rounds) followed by two parallel $\mathsf{SecMatMul}$ operations (1 round). The total is \textbf{3 rounds} with $O(d^2 k)$ communication.
\end{itemize}

\paragraph{End-to-End Complexity.}
The complete \name{} protocol (Algorithm~\ref{alg:verify}) proceeds layer-by-layer. For a target layer $l$, Algorithm~\ref{alg:lbp} iterates $l-1$ times. Each iteration involves: $\mathsf{SecMatMul}$ (1 round) $\to$ slope multiplication (1 round) $\to$ $\Pi_{\boldsymbol{\delta}}$ (3 rounds), yielding 5 rounds per iteration. The bias accumulation in Lines 6-7 overlaps with $\Pi_{\boldsymbol{\delta}}$: the $\mathsf{SecMatMul}(\langle \mathbf{A} \rangle, \langle \mathbf{b} \rangle)$ executes in parallel with the SecReLU operations, and the final addition of $\boldsymbol{\delta}$ is a local operation.

Summing over all layers, the total complexity is:
\begin{itemize}
    \item \textbf{Rounds:} $O(L^2)$. The iterative bound computation contributes $\sum_{l=2}^{L-1} 5(l-1) = \frac{5(L-2)(L-1)}{2} \approx \frac{5}{2}L^2$ rounds. The final margin certification adds $5(L-1) = O(L)$ rounds. Additionally, $\Pi_{\boldsymbol{\alpha}}$ is invoked $L-1$ times, contributing $(L-1)(7 + 2n_{\mathrm{iter}})$ rounds.
    \item \textbf{Communication:} $O(L^2 d^2 k)$. This is dominated by the dense matrix multiplications in the iterative backward pass.
    
    \item \textbf{Computation:} $O(L^2 d^3)$, dominated by $O(L^2)$ matrix multiplications of size $d \times d$.
    
    \item \textbf{Preprocessing:} $O(L^2 d^2 k)$ bits for matrix Beaver triples and $O(Ld k\lambda)$ bits for FSS keys.
\end{itemize}

\subsection{Error Analysis}
\label{appendix:error-analysis}

We analyze the LBP bound computation, which constructs an affine form bound by backward propagation of a coefficient matrix $\mathbf{A}$ and a bias vector $\mathbf{c}$, accumulating constants and relaxation terms layer by layer, and then forms the final upper/lower bound of the target quantity as a function of the input.
The final step is to maximize/minimize this affine form over the input perturbation set $\mathcal{X}$, which yields closed-form expressions for the bounds.

We denote by $r(\mathbf{M})$ the vector of row-wise sums of absolute values of a matrix $M$, i.e., $r(\mathbf{M})_i = \sum_j |\mathbf{M}_{ij}|$.
We denote by $\|\mathbf{M}\|_{\infty} = \max_i r(\mathbf{M})_i$ the induced infinity norm of a matrix $\mathbf{M}$, and by $\|\mathbf{v}\|_1 = \sum_i |v_i|$ the $\ell_1$ norm of a vector $\mathbf{v}$.
Let $Q_{n_f} : \mathbb{R} \to \mathbb{R}$ denote the fixed-point quantization function that maps a real number to its fixed-point representation with $n_f$ fractional bits, i.e., $Q_{n_f}(y) = 2^{-n_f} \lfloor 2^{n_f} y \rfloor$.

For each layer $i$, the algorithm uses a diagonal matrix of slopes $\mathrm{diag}(\bm{\alpha}^{(i)})$ with entries $\alpha_k^{(i)} \in [0,1]$ to relax the ReLU activation, applied during backward propagation.
For a fixed fractional precision parameter $n_f$, let $\mathrm{ARS}_{n_f}(\cdot)$ denote the arithmetic-right-shift scaling operation that divides its input by $2^{n_f}$ and rounds towards $-\infty$ (i.e., takes the floor).
\emph{We model each call to $\mathrm{ARS}_{n_f}$ as introducing an additive element-wise rounding error bounded by a constant $\varepsilon_q$ in real units}, where typically $\varepsilon_q = 2^{-n_f}$.

\paragraph{Assumptions}

Our analysis makes the following assumptions:
\begin{itemize}
    \item \textbf{Assumption 1 (Range safety)}: During execution, every intermediate integer value remains within the representable range of the fixed-point format, i.e., no overflow occurs. Equivalently, the fixed-point computation is a faithful representation of the corresponding real arithmetic plus truncation errors.
    \item \textbf{Assumption 2 (Quantization error model)}: Each call to $\mathrm{ARS}_{n_f}$ corresponds to applying $Q_{n_f}$ to the intended real-valued quantity, and thus introduces an additive error bounded by $\varepsilon_q$ in real units:
        \begin{equation*}
            Q_{n_f}(y) = y + \eta(y) \quad \text{with} \quad |\eta(y)| \leq \varepsilon_q.
        \end{equation*}
\end{itemize}F

\paragraph{Backward propagation operator}

For the target bound computation, define the ideal backward propagation of coefficient matrices $\{\mathbf{A}^{(i)}\}$ by
\begin{equation*}
    \mathbf{A}^{(i-1)} = \mathbf{A}^{(i)} \mathbf{W}^{(i)} \mathrm{diag}(\bm{\alpha}^{(i-1)}) , \quad i=1,\ldots,m,
\end{equation*}
where $m$ is the number of backward steps performed by the algorithm (e.g., $m=L-1$ for the full backward propagation from the last layer).
The starting coefficient $\mathbf{A}^{(m)}$ is initialized as the weight matrix of the current layer.

Let $\{\mathbf{\Lambda}^{(i)}\}$ be the corresponding coefficients computed by the fixed-point implementation, which applies the $\mathrm{ARS}_{n_f}$ operation at each multiplication step.
The coefficient error at layer $i$ is defined as $\mathbf{E}^{(i)} = \mathbf{\Lambda}^{(i)} - \mathbf{A}^{(i)}$.

\begin{theorem}\label{thm:coefficient_error}
    Under assumptions 1 and 2, given the slopes $\bm{\alpha}^{(i)} \in [0,1]$ element-wise, the coefficient error satisfies the recurrence
    \begin{equation}
        \| \mathbf{E}^{(i-1)} \|_\infty \leq \| \mathbf{W}^{(i)} \|_\infty \| \mathbf{E}^{(i)} \|_\infty + \Delta_i, \quad i = 1,\ldots,m,
    \end{equation}
    where $\Delta_i$ is an upper bound on the $\ell_\infty$ norm of the local rounding injection at step $i$, induced by the truncations performed at that step.
    In particular, if step $i$ produces a matrix with $q_i$ columns before truncation, then one may take the conservative bound $\Delta_i \leq q_i \varepsilon_q$.
    Consequently, if the initial coefficient is computed exactly (i.e., $\mathbf{E}^{(m)} = \mathbf{0}$), then the final coefficient error after $m$ steps is bounded by
    \begin{equation}
        \| \mathbf{E}^{(0)} \|_\infty \leq \sum_{i=1}^m \left( \prod_{j=1}^{i-1} \| \mathbf{W}^{(j)} \|_\infty \right) \Delta_i.
    \end{equation}
    If, furthermore, $\| \mathbf{W}^{(i)} \|_\infty \leq \rho$ for all relevant layers, then the error bound simplifies to
    \begin{equation}
        \| \mathbf{E}^{(0)} \|_\infty \leq \left( \sum_{k=0}^{m-1} \rho^k \right) \max_i \Delta_i =
        \begin{cases}
            O(m) \cdot \max_i \Delta_i, & \textit{if }\rho = 1, \\
            O(\rho^m) \cdot \max_i \Delta_i, & \textit{if }\rho > 1, \\
            O(1) \cdot \max_i \Delta_i, & \textit{if }\rho < 1.
        \end{cases}
    \end{equation}
\end{theorem}

\begin{proof}
    Fix $i \in \{1,\ldots,m\}$. The ideal update is $\mathbf{A}^{(i-1)} = \mathbf{A}^{(i)} \mathbf{W}^{(i)} \mathrm{diag}(\bm{\alpha}^{(i-1)})$, while the actual update is
    \begin{equation}
        \mathbf{\Lambda}^{(i-1)} = Q_{n_f}(\mathbf{\Lambda}^{(i)} \mathbf{W}^{(i)} \mathrm{diag}(\bm{\alpha}^{(i-1)})),
    \end{equation}
    where we have used the fact that $\alpha_k^{(i-1)}$ is computed to lie in $[0,1]$ and treated it as the same slope vector used in both the ideal and actual computations.
    If one wishes to include quantized $\alpha_k^{(i-1)}$ in the error analysis, it appears as an additional additive term.

    By Assumption 2, we have
    \begin{equation}
        \mathbf{\Lambda}^{(i-1)} = \mathbf{\Lambda}^{(i)} \mathbf{W}^{(i)} \mathrm{diag}(\bm{\alpha}^{(i-1)}) + \mathbf{\Xi}^{(i)},
    \end{equation}
    where $\mathbf{\Xi}^{(i)}$ is the rounding error matrix injected by the $\mathrm{ARS}_{n_f}$ operation, with $| \mathbf{\Xi}^{(i)}_{j,k}| \leq \varepsilon_q$ for all entries.
    Therefore,
    \begin{equation}
        \mathbf{E}^{(i-1)}
        = \mathbf{\Lambda}^{(i-1)} - \mathbf{A}^{(i-1)}
        = (\mathbf{\Lambda}^{(i)} - \mathbf{A}^{(i)}) \mathbf{W}^{(i)} \mathrm{diag}(\bm{\alpha}^{(i-1)})  + \mathbf{\Xi}^{(i)}
        = \mathbf{E}^{(i)} \mathbf{W}^{(i)} \mathrm{diag}(\bm{\alpha}^{(i-1)})  + \mathbf{\Xi}^{(i)}.
    \end{equation}
    Taking $\| \cdot \|_\infty$ on both sides and using the sub-multiplicative property gives
    \begin{equation}
        \| \mathbf{E}^{(i-1)} \|_\infty \leq \| \mathbf{E}^{(i)} \|_\infty \| \mathbf{W}^{(i)} \|_\infty \| \mathrm{diag}(\bm{\alpha}^{(i-1)}) \|_\infty + \| \mathbf{\Xi}^{(i)} \|_\infty.
    \end{equation}
    Since $\alpha_k^{(i-1)} \in [0,1]$ for all $k$, we have $\| \mathrm{diag}(\bm{\alpha}^{(i-1)}) \|_\infty = \max_k \alpha_k^{(i-1)} \leq 1$, hence
    \begin{equation}
        \| \mathbf{E}^{(i-1)} \|_\infty \leq \| \mathbf{W}^{(i)} \|_\infty \| \mathbf{E}^{(i)} \|_\infty + \| \mathbf{\Xi}^{(i)} \|_\infty.
    \end{equation}
    Define $\Delta_i = \| \mathbf{\Xi}^{(i)} \|_\infty$ to be the local error bound at step $i$, yielding the stated recurrence.

    To bound $\Delta_i$, note that $\| \mathbf{\Xi}^{(i)} \|_\infty = \max_j \sum_k |\mathbf{\Xi}^{(i)}_{j,k}| \leq \max_j \sum_k \varepsilon_q = q_i \varepsilon_q$, where $q_i$ is the number of columns of the matrix.

    Finally, unrolling the recurrence with $\mathbf{E}^{(m)} = \mathbf{0}$ gives the stated bound on $\mathbf{E}^{(0)}$:
    \begin{align*}
        \| \mathbf{E}^{(0)} \|_\infty & \leq \| \mathbf{W}^{(1)} \|_\infty \| \mathbf{E}^{(1)} \|_\infty + \Delta_1 \\
        & \leq \| \mathbf{W}^{(1)} \|_\infty ( \| \mathbf{W}^{(2)} \|_\infty \| \mathbf{E}^{(2)} \|_\infty + \Delta_2 ) + \Delta_1 \\
        & = \| \mathbf{W}^{(1)} \|_\infty \| \mathbf{W}^{(2)} \|_\infty \| \mathbf{E}^{(2)} \|_\infty +  \| \mathbf{W}^{(1)}\|_\infty  \Delta_2 +  \Delta_1  \\
        & \cdots \\
        & \leq \sum_{i=1}^m \left( \prod_{j=1}^{i-1} \| \mathbf{W}^{(j)} \|_\infty \right) \Delta_i.
    \end{align*}
    The $\rho$-based bound follows directly from the above by substituting $\| \mathbf{W}^{(j)} \|_\infty \leq \rho$ for all $j$.
\end{proof}

Based on Theorem~\ref{thm:coefficient_error}, we now derive the error bound for the final upper/lower bound vectors.
Write the ideal final bounds over $\mathcal{X}$ as
\begin{equation*}
    UB = \mathbf{A}^{(0)} \mathbf{x}_0 + \ub{\mathbf{c}}^{(0)} + \epsilon \cdot r(\mathbf{A}^{(0)}),
    \quad
    LB = \mathbf{A}^{(0)} \mathbf{x}_0 + \lb{\mathbf{c}}^{(0)} - \epsilon \cdot r(\mathbf{A}^{(0)}),
\end{equation*}
where $r(\mathbf{A}^{(0)})$ is the row-wise $\ell_1$-norm vector of $\mathbf{A}^{(0)}$.

During backward propagation, the ideal update is
\begin{align}
    \mathbf{\ub{c}}^{(i-1)} & = \mathbf{\ub{c}}^{(i)} + \mathbf{A}^{(i)} \mathbf{b}^{(i)} + \bm{\ub{\delta}}^{(i-1)}, \\
    \mathbf{\lb{c}}^{(i-1)} & = \mathbf{\lb{c}}^{(i)} + \mathbf{A}^{(i)} \mathbf{b}^{(i)} + \bm{\lb{\delta}}^{(i-1)},
\end{align}
where $\bm{\ub{\delta}}^{(i-1)}$ and $\bm{\lb{\delta}}^{(i-1)}$ are the ReLU relaxation intercept contribution terms.
Define $u_k^{(i-1)} = (-\lb{z}^{(i-1)}_k)_+ \ge 0$ for all $k$, where $(x)_+ = \max(0, x)$ denotes the positive part, with the convention $u_k^{(0)} = 0$ since layer $0$ is the input layer with no ReLU.
Then the intercept contributions are:
\begin{align}
    \bm{\ub{\delta}}^{(i-1)}_j & = \sum_{k} (\mathbf{A}^{(i-1)}_{j,k})_+ u_k^{(i-1)}, \\
    \bm{\lb{\delta}}^{(i-1)}_j & = \sum_{k} (-\mathbf{A}^{(i-1)}_{j,k})_+ u_k^{(i-1)}.
\end{align}
The actual update performed by the fixed-point implementation is
\begin{align}
    \bm{\ub{\gamma}}^{(i-1)} & = Q_{n_f}(\bm{\ub{\gamma}}^{(i)} + \mathbf{\Lambda}^{(i)} \mathbf{b}^{(i)} + \hat{\bm{\ub{\delta}}}^{(i-1)}), \\
    \bm{\lb{\gamma}}^{(i-1)} & = Q_{n_f}(\bm{\lb{\gamma}}^{(i)} + \mathbf{\Lambda}^{(i)} \mathbf{b}^{(i)} + \hat{\bm{\lb{\delta}}}^{(i-1)}),
\end{align}
where $\hat{\bm{\ub{\delta}}}^{(i-1)}$ and $\hat{\bm{\lb{\delta}}}^{(i-1)}$ are computed using the actual coefficient matrix $\mathbf{\Lambda}^{(i-1)}$:
\begin{align}
    \hat{\bm{\ub{\delta}}}^{(i-1)}_j & = \sum_{k} (\mathbf{\Lambda}^{(i-1)}_{j,k})_+ u_k^{(i-1)}, \\
    \hat{\bm{\lb{\delta}}}^{(i-1)}_j & = \sum_{k} (-\mathbf{\Lambda}^{(i-1)}_{j,k})_+ u_k^{(i-1)}.
\end{align}
Therefore, the error in the bias term at layer $i-1$, denoted as $\bm{\ub{\nu}}^{(i-1)} \coloneq \bm{\ub{\gamma}}^{(i-1)} - \mathbf{\ub{c}}^{(i-1)}$, can be decomposed component-wise (for index $j$) as
\begin{align}
    \bm{\ub{\nu}}^{(i-1)}_j
    &= \bm{\ub{\gamma}}^{(i)}_j + [\mathbf{\Lambda}^{(i)} \mathbf{b}^{(i)}]_j + \hat{\bm{\ub{\delta}}}^{(i-1)}_j + \Theta^{(i)}_j
    - \Bigl(\mathbf{\ub{c}}^{(i)}_j + [\mathbf{A}^{(i)} \mathbf{b}^{(i)}]_j + \bm{\ub{\delta}}^{(i-1)}_j\Bigr) \notag \\
    &= \bm{\ub{\nu}}^{(i)}_j + [\mathbf{E}^{(i)} \mathbf{b}^{(i)}]_j
    + \sum_{k} \Bigl[(\mathbf{\Lambda}^{(i-1)}_{j,k})_+ - (\mathbf{A}^{(i-1)}_{j,k})_+\Bigr] u_k^{(i-1)}
    + \Theta^{(i)}_j,
\end{align}
where $\Theta^{(i)}$ is the rounding error injected by the $\mathrm{ARS}_{n_f}$ operation, with $\|\Theta^{(i)}\|_\infty \leq \varepsilon_q$.
Using the 1-Lipschitz property of $(\cdot)_+$, we can bound the error in the intercept contribution:
\begin{equation}
    |\hat{\bm{\ub{\delta}}}^{(i-1)}_j - \bm{\ub{\delta}}^{(i-1)}_j| \le \sum_{k} |(\mathbf{\Lambda}^{(i-1)}_{j,k})_+ - (\mathbf{A}^{(i-1)}_{j,k})_+| u_k^{(i-1)} \le \sum_{k} |\mathbf{E}^{(i-1)}_{j,k}| u_k^{(i-1)} \le \|\mathbf{E}^{(i-1)}_{j,:}\|_1 \|u^{(i-1)}\|_\infty.
\end{equation}
Taking the infinity norm over $j$, we get:
\begin{equation}
    \|\hat{\bm{\ub{\delta}}}^{(i-1)} - \bm{\ub{\delta}}^{(i-1)}\|_\infty \le \|\mathbf{E}^{(i-1)}\|_\infty \|u^{(i-1)}\|_\infty.
\end{equation}
Unrolling the recurrence for $\bm{\ub{\nu}}^{(0)}$ (with $\bm{\ub{\nu}}^{(m)} = \mathbf{0}$) gives
\begin{equation}
    \|\bm{\ub{\nu}}^{(0)}\|_\infty = \|\bm{\ub{\gamma}}^{(0)} - \mathbf{\ub{c}}^{(0)}\|_\infty \leq \sum_{i=1}^m \left( \|\mathbf{E}^{(i)}\|_\infty \|\mathbf{b}^{(i)}\|_\infty + \|\mathbf{E}^{(i-1)}\|_\infty \|u^{(i-1)}\|_\infty \right) + m \varepsilon_q.
\end{equation}
A symmetric bound holds for the lower bound bias error $\|\bm{\lb{\nu}}^{(0)}\|_\infty = \|\bm{\lb{\gamma}}^{(0)} - \mathbf{\lb{c}}^{(0)}\|_\infty$ (replacing $(\cdot)_+$ by $(-\cdot)_+$ throughout).

\begin{theorem}\label{thm:bound_error}
    Let $\widehat{UB},\widehat{LB}$ be the fixed-point outputs, and define $\mathbf{E}^{(0)} = \mathbf{\Lambda}^{(0)} - \mathbf{A}^{(0)}$.
    Under Assumptions 1 and 2, the absolute errors of the final bounds are bounded by:
    \begin{align}
        \|\widehat{UB} - UB\|_\infty
        &\le (\|\mathbf{x}_0\|_\infty + \epsilon)\|\mathbf{E}^{(0)}\|_\infty
        + \sum_{i=1}^m \left( \|\mathbf{E}^{(i)}\|_\infty \|\mathbf{b}^{(i)}\|_\infty + \|\mathbf{E}^{(i-1)}\|_\infty \|u^{(i-1)}\|_\infty \right) + \zeta_u, \\
        \|\widehat{LB} - LB\|_\infty
        &\le (\|\mathbf{x}_0\|_\infty + \epsilon)\|\mathbf{E}^{(0)}\|_\infty
        + \sum_{i=1}^m \left( \|\mathbf{E}^{(i)}\|_\infty \|\mathbf{b}^{(i)}\|_\infty + \|\mathbf{E}^{(i-1)}\|_\infty \|u^{(i-1)}\|_\infty \right) + \zeta_l,
    \end{align}
    where $u^{(i-1)} = (-\lb{\mathbf{z}}^{(i-1)})_+ \ge 0$ with $u^{(0)} = 0$ (no ReLU at the input layer), and $\zeta_u, \zeta_l$ are assembly-level quantization residuals bounded by $(m+1)\varepsilon_q$.
    Therefore, by substituting Theorem~\ref{thm:coefficient_error} bounds for $\|\mathbf{E}^{(i)}\|_\infty$, one obtains an explicit depth-dependent absolute error bound for both $\widehat{UB}$ and $\widehat{LB}$.
\end{theorem}

\begin{proof}
    We prove the upper-bound case; the lower-bound case is identical.
    The ideal final upper bound is given by
    \begin{equation*}
        UB = \mathbf{A}^{(0)} \mathbf{x}_0 + \mathbf{\ub{c}}^{(0)} + \epsilon \cdot r(\mathbf{A}^{(0)}).
    \end{equation*}
    The fixed-point implementation computes the final bound as
    \begin{equation*}
        \widehat{UB} = \mathbf{\Lambda}^{(0)} \mathbf{x}_0 + \bm{\ub{\gamma}}^{(0)} + \epsilon \cdot r(\mathbf{\Lambda}^{(0)}) + \xi_u,
    \end{equation*}
    where $\xi_u$ represents the rounding error introduced during the final assembly operations (e.g., dot products and additions), with $\|\xi_u\|_\infty \le \varepsilon_q$.
    Decomposing the difference yields
    \begin{equation*}
        \widehat{UB} - UB
        = \mathbf{E}^{(0)} \mathbf{x}_0
        + \epsilon\bigl(r(\mathbf{\Lambda}^{(0)})-r(\mathbf{A}^{(0)})\bigr)
        + \bm{\ub{\nu}}^{(0)}
        + \xi_u.
    \end{equation*}
    Taking the infinity norm and applying the triangle inequality gives
    \begin{equation*}
        \|\widehat{UB}-UB\|_\infty
        \le \|\mathbf{E}^{(0)}\mathbf{x}_0\|_\infty
        + \epsilon\|r(\mathbf{\Lambda}^{(0)})-r(\mathbf{A}^{(0)})\|_\infty
        + \|\bm{\ub{\nu}}^{(0)}\|_\infty
        + \|\xi_u\|_\infty.
    \end{equation*}
    By norm inequalities, we have
    \begin{equation*}
        \|\mathbf{E}^{(0)}\mathbf{x}_0\|_\infty \le \|\mathbf{E}^{(0)}\|_\infty\|\mathbf{x}_0\|_\infty,
        \quad
        \|r(\mathbf{\Lambda}^{(0)})-r(\mathbf{A}^{(0)})\|_\infty \le \|\mathbf{E}^{(0)}\|_\infty.
    \end{equation*}
    Substituting the bound for $\|\bm{\ub{\nu}}^{(0)}\|_\infty$ derived earlier, we obtain
    \begin{align*}
        \|\widehat{UB}-UB\|_\infty
        &\le (\|\mathbf{x}_0\|_\infty + \epsilon)\|\mathbf{E}^{(0)}\|_\infty
        + \sum_{i=1}^m \left( \|\mathbf{E}^{(i)}\|_\infty \|\mathbf{b}^{(i)}\|_\infty + \|\mathbf{E}^{(i-1)}\|_\infty \|u^{(i-1)}\|_\infty \right) + m \varepsilon_q + \varepsilon_q \\
        &= (\|\mathbf{x}_0\|_\infty + \epsilon)\|\mathbf{E}^{(0)}\|_\infty
        + \sum_{i=1}^m \left( \|\mathbf{E}^{(i)}\|_\infty \|\mathbf{b}^{(i)}\|_\infty + \|\mathbf{E}^{(i-1)}\|_\infty \|u^{(i-1)}\|_\infty \right) + \zeta_u,
    \end{align*}
    where $\zeta_u \le (m+1)\varepsilon_q$.
\end{proof}

\begin{remark}
{The stability constant $\epsilon_s$ in Eq.~\eqref{eq:alpha_unified}
perturbs each relaxation slope from its ideal value, and hence the
relaxation bounding functions. For every neuron $k$, the sup-norm
deviation of the perturbed bounding functions over
$[\underline{z}_k, \bar{z}_k]$ is at most $\epsilon_s$: for
unstable neurons this follows from
$\max(|\underline{z}_k|, \bar{z}_k) \leq \bar{z}_k - \underline{z}_k$;
the active case gives deviation
$\epsilon_s \bar{z}_k / (\bar{z}_k + \epsilon_s) \leq \epsilon_s$,
and the inactive case is exact. 
Substituting the perturbed relaxations in the backward recursion
therefore injects, at each step $i$, an additional deviation of at
most $\|\hat{\mathbf{A}}^{(i)}\|_\infty\,\epsilon_s$ into the
accumulated bound values, where $\hat{\mathbf{A}}^{(i)}$ is the
intermediate coefficient matrix in Eq.~\eqref{eq:Ahat_update} and,
as in the proof of Theorem~\ref{thm:coefficient_error}, the
relaxation parameters at the remaining layers are treated as
fixed. The effect on the final bounds thus follows the same
depth-dependent structure as Theorem~\ref{thm:bound_error}, with
the per-step injection scale $\varepsilon_q$ replaced by
$\|\hat{\mathbf{A}}^{(i)}\|_\infty\,\epsilon_s$.
Since the perturbed relaxation deviates from the ideal one, a
verification outcome could in principle change when the certified
margin lies within $O(\epsilon_s)$ of zero; empirically, no such
change is observed even at $1000\times$ the default $\epsilon_s$
(Table~\ref{tab:ablation_eps}). A slope--intercept co-adjustment
that restores provable soundness of the perturbed relaxation is
left for future work.}
\end{remark}

\begin{corollary}[Asymptotic error bound]
    Under the hypotheses of Theorems~\ref{thm:coefficient_error} and~\ref{thm:bound_error}, assume $\mathbf{E}^{(m)}=\mathbf{0}$, $\|\mathbf{W}^{(i)}\|_\infty\le\rho$ and $\Delta_i\le\Delta$ for all $i$.
    Let
    \[
        X \;=\; \|\mathbf{x}_0\|_\infty+\epsilon,
        \qquad
        B \;=\; \max_i \|\mathbf{b}^{(i)}\|_\infty,
        \qquad
        U \;=\; \max_i \|u^{(i-1)}\|_\infty,
    \]
    and let $S_n = \sum_{k=0}^{n-1}\rho^k$, so that the layer-$t$ coefficient error satisfies $\|\mathbf{E}^{(t)}\|_\infty \le S_{m-t}\,\Delta$.
    Then, as $m\to\infty$, the error in $\widehat{UB}$ (and identically in $\widehat{LB}$) satisfies
    \begin{equation}
        \|\widehat{UB}-UB\|_\infty
        \;\le\;
        \Bigl[X\cdot S_m \;+\; B\sum_{i=1}^{m} S_{m-i} \;+\; U\sum_{i=1}^{m} S_{m-i+1}\Bigr]\Delta \;+\; O(m)\,\varepsilon_q,
    \end{equation}
    where the three asymptotic regimes of the leading term (as a function of $m$) are:
    \begin{equation}
        \|\widehat{UB}-UB\|_\infty \;=\;
        \begin{cases}
            O(m)\cdot(B+U)\Delta \;+\; O(1)\cdot X\Delta \;+\; O(m)\,\varepsilon_q,
            & \text{if } \rho < 1, \\[4pt]
            O(m^2)\cdot(B+U)\Delta \;+\; O(m)\cdot X\Delta \;+\; O(m)\,\varepsilon_q,
            & \text{if } \rho = 1, \\[4pt]
            O(\rho^m)\cdot(X+B+U)\Delta \;+\; O(m)\,\varepsilon_q,
            & \text{if } \rho > 1.
        \end{cases}
    \end{equation}
    In particular, the $\zeta$ residuals are $O(m)\,\varepsilon_q$ in all cases.
\end{corollary}

\paragraph{Interpretation of the unstable regime.}
The $\rho>1$ branch in the asymptotic corollary should be read as a
conservative worst-case envelope rather than as the typical numerical
behavior of \name{}.
It is obtained by replacing every layer-dependent amplification factor
with the same uniform upper bound $\rho$ and every local rounding
injection with the same upper bound $\Delta$.
The non-asymptotic analysis above is sharper and remains layer-wise.
Indeed, before the conservative inequality
$\|\mathrm{diag}(\bm{\alpha}^{(i-1)})\|_\infty\le 1$ is applied, the
proof of Theorem~\ref{thm:coefficient_error} gives
\begin{equation}
    \|\mathbf{E}^{(i-1)}\|_\infty
    \le
    \kappa_i \|\mathbf{E}^{(i)}\|_\infty + \Delta_i,
    \qquad
    \kappa_i
    \coloneq
    \|\mathbf{W}^{(i)}\mathrm{diag}(\bm{\alpha}^{(i-1)})\|_\infty .
\end{equation}
Thus, when $\mathbf{E}^{(m)}=\mathbf{0}$,
\begin{equation}
    \|\mathbf{E}^{(0)}\|_\infty
    \le
    \sum_{i=1}^{m}
    \left(\prod_{j=1}^{i-1}\kappa_j\right)\Delta_i .
\end{equation}
This expression propagates each local error through the actual product
of effective layer norms.
Consequently, the simplified $O(\rho^m)$ term arises only when these
heterogeneous products are upper-bounded by the same worst-case factor
at every layer; factors $\kappa_j<1$ can attenuate errors produced at
other layers with $\kappa_j>1$.
The diagonal relaxation matrices further reduce the effective
amplification: since $\alpha_k^{(i-1)}\in[0,1]$, the factor
$\kappa_i$ is a row-wise weighted sum of $|\mathbf{W}^{(i)}|$ with
weights $\bm{\alpha}^{(i-1)}$, and is always at most
$\|\mathbf{W}^{(i)}\|_\infty$.
When many slopes are strictly below one, as for inactive or unstable
ReLU neurons, $\kappa_i$ can be substantially smaller than the raw
weight norm of the layer and can even fall below one when
$\|\mathbf{W}^{(i)}\|_\infty>1$.

We do not rule out the vulnerability captured by this worst-case bound:
if a network has consistently large effective layer norms, fixed-point
rounding errors can be amplified and the numerical precision of the
certified bounds can deteriorate.
In standard modern training pipelines, however, small initialization,
weight decay, and related regularization or normalization techniques
typically limit uncontrolled growth of weight norms.
Consistent with this interpretation, our experiments do not exhibit the
exponential-growth regime: \name{} maintains 100\% verification
consistency with plaintext CROWN across all evaluated architectures and
perturbation settings, including near-boundary instances
(Table~\ref{tab:fidelity_final_complete}).

\section{{Discussion on Malicious Security}}
\label{app:malicious}

{SECURECROWN currently targets the semi-honest model, which is 
standard for initial 2PC-based protocol designs. A malicious-secure 
upgrade could replace unauthenticated Beaver triples with authenticated 
triples and add MAC-based consistency checks, following SPDZ-style 
protocols. The trusted-dealer assumption can also be removed by using 
two-party preprocessing, such as OT-based triple generation, without 
changing the online verifier logic. These changes preserve the 
branch-free reformulation of CROWN, but increase both online 
communication and offline preprocessing costs.}

{To estimate this overhead, we implemented the same verification 
algorithm in MP-SPDZ. On MNIST \(5{\times}[256]\), replacing 
\texttt{semi2k} with \texttt{spdz2k} increases online time by about 
\(3.2\times\) (\(183\)s to \(586\)s) and communication by about 
\(8.7\times\) (\(1.26\)TB to \(10.88\)TB). In comparison, our 
FSS-based semi-honest implementation completes the same task in 
\(4.3\)s with \(113\)MB communication. This suggests that a specialized 
malicious-secure variant of SECURECROWN could be substantially more 
efficient than directly using a generic malicious 2PC backend, although 
designing and benchmarking such a variant is left for future work.}

\section{{Extension to Other Network Architectures}}
\label{app:extensions}

{The current implementation targets CROWN-style LBP for fully 
connected ReLU networks. We discuss extensions to other network 
components according to the type of additional support required.}

\subsection{{Components Requiring Minimal Cryptographic Changes}}

{\textbf{Batch normalization.} Batch normalization (BN) can be 
folded into the preceding affine layer by \(P_0\) before secret 
sharing, since all BN parameters are fixed after training. This 
requires no change to the secure verification protocol.}

{\textbf{Residual connections.} Residual connections compute 
\(\mathbf{y}=F(\mathbf{x})+\mathbf{x}\), where \(F(\mathbf{x})\) 
denotes the residual branch. From the secure-computation perspective, 
this addition is local under additive secret sharing and incurs no 
communication. Therefore, residual connections require no new 
cryptographic primitives, although the verifier implementation would 
need to handle skip connections and merge the propagated bounds from 
multiple graph branches.}

\subsection{{Linear Components with Scalability Bottlenecks}}

{\textbf{Convolutional layers.} Convolutions can be represented as 
matrix multiplications 
\(\mathbf{y}=W_{\mathrm{toep}}\mathbf{x}+\mathbf{b}\), where 
\(\mathbf{x}\) and \(\mathbf{y}\) denote vectorized input and output 
feature maps, and \(W_{\mathrm{toep}}\) is the sparse Toeplitz matrix 
induced by the convolutional kernels. Therefore, convolutions are 
compatible with our \texttt{SecMatMul} interface in principle. However, 
a naive dense unfolding loses the sparsity of the convolutional kernel. 
For input and output channel counts \(c_{\mathrm{in}}\) and 
\(c_{\mathrm{out}}\), input and output spatial sizes \(s\) and \(s'\), 
and kernel size \(\kappa\), the unfolded matrix has dimension 
\((c_{\mathrm{out}}s'^2)\times(c_{\mathrm{in}}s^2)\), whereas the 
convolution itself has only \(c_{\mathrm{out}}c_{\mathrm{in}}\kappa^2\) 
independent kernel parameters. Since our Beaver-based \texttt{SecMatMul} 
treats the unfolded matrix as dense, communication scales with the full 
matrix dimensions rather than with the number of independent kernel 
parameters. Practical support for CNNs would therefore require secure 
convolution protocols that exploit kernel sparsity or tiled/sliding-window 
structure.}

{\textbf{Cost projection for ResNet-18.} For a CIFAR-10-scale 
ResNet-18, a naive Toeplitz unfolding of a \(64\)-channel 
\(32\times32\) convolution can yield matrices up to 
\(65{,}536\times65{,}536\). Repeated dense secure matrix multiplications 
inside LBP would therefore lead to communication on the order of 
terabytes over the full verification procedure. Exploiting the sparse 
Toeplitz structure, as done in optimized secure-inference systems such 
as CryptGPU~\citep{Tan2021cryptgpu} and 
CrypTFlow~\citep{Kumar2020cryptflow}, is a promising direction for 
reducing this cost.}

\subsection{{Components Requiring New Verification Relaxations}}

{\textbf{Non-ReLU activations.} Non-ReLU activations, such as sigmoid 
or tanh, require activation-specific linear relaxations. Extending 
SecureCROWN to these activations would therefore require new branch-free 
reformulations of the corresponding relaxation logic. This does not 
necessarily require new cryptographic primitives: the underlying FSS 
primitives already support operations such as spline and exponential 
evaluations~\citep{Gupta2025shark}, providing a cryptographic foundation 
for these extensions. The main challenge is deriving 
secure-computation-friendly bound propagation rules that minimize secure 
comparisons and communication overhead.}

{\textbf{Attention layers.} Attention layers compute 
\(\mathrm{softmax}(QK^\top/\sqrt{d_k})V\), where \(Q\), \(K\), and 
\(V\) denote the query, key, and value matrices, respectively, and 
\(d_k\) is the key dimension. This operation involves bilinear terms in 
\(QK^\top\) and the softmax nonlinearity. Unlike ReLU, which is a 
univariate piecewise-linear function, attention introduces nonlinearities 
whose tight and scalable verification remains substantially more complex 
even in plaintext verification~\citep{shi2025genbab}. Extending 
SECURECROWN to attention-based architectures would therefore first 
require mature CROWN-style relaxations for these operations.}

\end{document}